\documentclass[sigconf]{acmart}

\renewcommand\footnotetextcopyrightpermission[1]{} 
\usepackage{amsmath,amssymb,amsthm,amsfonts}
\usepackage{xspace}
\usepackage{graphicx}
\usepackage{ifpdf}
\usepackage{longtable}
\usepackage{tabu}
\usepackage{mathtools}

\ifpdf
    \usepackage{libertine}
    \usepackage[T1]{fontenc}
	\usepackage[utf8]{inputenc}
    \usepackage{hyperref}
    \hypersetup{
      colorlinks=true,
      linkcolor=blue,
      linktoc=all,
      citecolor=green!30!black,
      bookmarksnumbered=true
    }
\fi

\usepackage{algorithm}
\usepackage[noend]{algpseudocode}
\usepackage{cleveref}
\usepackage{tcolorbox}
\usepackage{perpage}

\newlength\abovesectionskip
\newlength\belowsectionskip
\def\sectionfont{\normalfont\Large\bfseries}
\newlength\abovesubsectionskip
\newlength\belowsubsectionskip
\def\subsectionfont{\normalfont\large\bfseries}
\newlength\abovesubsubsectionskip
\newlength\belowsubsubsectionskip
\def\subsubsectionfont{\normalfont\normalsize\bfseries}
\newlength\aboveparagraphskip
\newlength\belowparagraphskip
\def\paragraphfont{\normalfont\normalsize\bfseries}
\makeatletter
\def\section{\@startsection{section}{1}{\z@}{-\abovesectionskip}%
               {\belowsectionskip}{\sectionfont}}
\def\subsection{\@startsection{subsection}{2}{\z@}{-\abovesubsectionskip}%
                  {\belowsubsectionskip}{\subsectionfont}}
\def\subsubsection{\@startsection{subsubsection}{3}{\z@}%
                     {-\abovesubsubsectionskip}{\belowsubsubsectionskip}%
                     {\subsubsectionfont}}
\def\paragraph{\@startsection{paragraph}{4}{\z@}{-\aboveparagraphskip}%
                 {-\belowparagraphskip}{\paragraphfont}}
\makeatother

\makeatletter
\renewenvironment{align*}{%
  \abovedisplayskip 5pt plus 1pt%
  \belowdisplayskip 5pt plus 1pt%
  \start@align\@ne\st@rredtrue\m@ne
}{%
  \endalign
}
\let\stdequation\equation
\renewcommand*\equation{%
  \abovedisplayskip 5pt plus 1pt%
  \belowdisplayskip 5pt plus 1pt%
  \stdequation}
\DeclareRobustCommand{\[}{
  \abovedisplayskip 5pt plus 1pt%
  \belowdisplayskip 5pt plus 1pt%
  \begin{equation*}
}
\makeatother

\usepackage{xcolor}
\usepackage{tikz}
\usetikzlibrary{calc}
\usetikzlibrary{patterns}

\makeatletter
\tikzset{%
  remember picture with id/.style={%
    remember picture,
    overlay,
    save picture id=#1,
  },
  save picture id/.code={%
    \edef\pgf@temp{#1}%
    \immediate\write\pgfutil@auxout{%
      \noexpand\savepointas{\pgf@temp}{\pgfpictureid}}%
  },
  if picture id/.code args={#1#2#3}{%
    \@ifundefined{save@pt@#1}{%
      \pgfkeysalso{#3}%
    }{
      \pgfkeysalso{#2}%
    }
  }
}

\def\savepointas#1#2{%
  \expandafter\gdef\csname save@pt@#1\endcsname{#2}%
}

\def\tmk@labeldef#1,#2\@nil{%
  \def\tmk@label{#1}%
  \def\tmk@def{#2}%
}

\tikzdeclarecoordinatesystem{pic}{%
  \pgfutil@in@,{#1}%
  \ifpgfutil@in@%
    \tmk@labeldef#1\@nil
  \else
    \tmk@labeldef#1,(0pt,0pt)\@nil
  \fi
  \@ifundefined{save@pt@\tmk@label}{%
    \tikz@scan@one@point\pgfutil@firstofone\tmk@def
  }{%
  \pgfsys@getposition{\csname save@pt@\tmk@label\endcsname}\save@orig@pic%
  \pgfsys@getposition{\pgfpictureid}\save@this@pic%
  \pgf@process{\pgfpointorigin\save@this@pic}%
  \pgf@xa=\pgf@x
  \pgf@ya=\pgf@y
  \pgf@process{\pgfpointorigin\save@orig@pic}%
  \advance\pgf@x by -\pgf@xa
  \advance\pgf@y by -\pgf@ya
  }%
}

\definecolor{mycolor}{HTML}{7DE7FE}
\definecolor{mycolor2}{HTML}{80F0FE}

\pgfdeclarepatternformonly{my
dots}{\pgfqpoint{-1pt}{-1pt}}{\pgfqpoint{2.5pt}{2.5pt}}{\pgfqpoint{1.5pt}{1.5pt}}%
{
  \pgfpathcircle{\pgfqpoint{0pt}{0pt}}{.35pt}
  \pgfpathcircle{\pgfqpoint{.75pt}{.75pt}}{.35pt}
  \pgfusepath{fill}
}

\newlength\AlgIndent
\newcounter{mymark}
\newcommand\ColorLine{%
  \stepcounter{mymark}%
  \tikz[remember picture with id=mark-\themymark,overlay] {;}%
  \begin{tikzpicture}[remember picture,overlay]%
    \filldraw[mycolor2]%
   let \p1=(pic cs:mark-\themymark), 
   \p2=(current page.east)  in 
   ([xshift=-\ALG@thistlm-0.03em,yshift=-0.7ex]0,\y1)  rectangle ++(\linewidth+\AlgIndent,\baselineskip);
  \end{tikzpicture}%
}%

\newcommand\ColorLineT{%
  \stepcounter{mymark}%
  \tikz[remember picture with id=mark-\themymark,overlay] {;}%
  \begin{tikzpicture}[remember picture,overlay]%
    \filldraw[mycolor2]%
   let \p1=(pic cs:mark-\themymark), 
   \p2=(current page.east)  in 
   ([xshift=-\ALG@thistlm-0.03em,yshift=-0.7ex]0,\y1-\baselineskip-0.1em)  rectangle ++(\linewidth+\AlgIndent,2.1*\baselineskip);
  \end{tikzpicture}%
}%

\algnewcommand\CREQUIRE{\item[\setlength\AlgIndent{1.6em}\ColorLine\algorithmicrequire]}%
\algnewcommand\CENSURE{\item[\setlength\AlgIndent{1.6em}\ColorLine\algorithmicensure]}%
\algnewcommand\CSTATE{\State\ColorLine}%
\algnewcommand\CSTATET{\State\ColorLineT}%
\algnewcommand\CSTATEx{\Statex\ColorLine}%
\algnewcommand\CCOMMENT{\Comment\ColorLine}%

\algdef{SE}[WHILE]{CWHILE}{ENDWHILE}%
   [2][default]{\ColorLine\algorithmicwhile\ #2\ \algorithmicdo}%
   {\algorithmicend\ \algorithmicwhile}%
\algdef{SE}[FOR]{CFOR}{ENDFOR}%
   [2][default]{\ColorLine\algorithmicfor\ #2\ \algorithmicdo}%
   {\algorithmicend\ \algorithmicfor}%
\algdef{S}[FOR]{CFORALL}%
   [2][default]{\ColorLine\algorithmicforall\ #2\ \algorithmicdo\ALG@compatcomm{#1}}%
\algdef{SE}[LOOP]{CLOOP}{ENDLOOP}%
   [1][default]{\ColorLine\algorithmicloop\ALG@compatcomm{#1}}%
   {\algorithmicend\ \algorithmicloop}%
\algdef{SE}[REPEAT]{CREPEAT}{UNTIL}%
   [1][default]{\ColorLine\algorithmicrepeat\ALG@compatcomm{#1}}%
   [1]{\algorithmicuntil\ #1}%
\algdef{SE}[IF]{CIF}{ENDIF}%
   [1]{\ColorLine\algorithmicif\ #1\ \algorithmicthen}%
   {\algorithmicend\ \algorithmicif}%
\algdef{SE}[IF]{CIFT}{ENDIF}%
   [1]{\ColorLineT\algorithmicif\ #1\ \algorithmicthen}%
   {\algorithmicend\ \algorithmicif}%
   
\algdef{C}[IF]{IF}{CELSIF}%
   [2][default]{\ColorLine\algorithmicelse\ \algorithmicif\ #2\ \algorithmicthen\ALG@compatcomm{#1}}%
\algdef{Ce}[ELSE]{IF}{CELSE}{ENDIF}%
   [1][default]{\ColorLine\algorithmicelse\ALG@compatcomm{#1}}%
   
\algtext*{ENDFOR}
\algtext*{ENDIF}
\algtext*{ENDWHILE}
\makeatother

\MakePerPage{footnote}

    \newtheoremstyle{mythmstyle}
      {6pt}   %
      {6pt}   
      {}            %
      {}            %
      {\bfseries}   %
      {. }          %
      {2.5pt}       %
      {\thmname{#1}\thmnumber{ #2}\thmnote{ \normalfont (#3)}}   %
    \theoremstyle{mythmstyle}
    \newtheorem{theorem}{Theorem}[section]\numberwithin{equation}{section}
    \newtheorem{corollary}[theorem]{Corollary}
    \newtheorem{claim}[theorem]{Claim}
    
    \newtheorem{lemma}[theorem]{Lemma}

    \newtheorem{fact}[theorem]{Fact}
    \newtheorem{remark}[theorem]{Remark}
    \newcommand{\abs}[1]{\lvert #1 \rvert}
    \newcommand{\card}[1]{\abs{#1}}
    \newcommand{\set}[1]{\left \{ #1 \right \}}                     
    \newcommand{\setst}[2]{\left\{\; #1 \,:\, #2 \;\right\}}        
    \newcommand{\union}{\cup}                                       
    \newcommand{\Union}{\bigcup}
    \newcommand{\intersect}{\cap}                                   
    
    \newcommand{\ceil}[1]{\left\lceil #1 \right\rceil}
    \newcommand{\argmax}{\operatornamewithlimits{arg\,max}}
    \newcommand{\bR}{\mathbb{R}}
    \newcommand{\cE}{\mathcal{E}}
    \newcommand{\tO}{\tilde{O}}
    
    \newcommand{\Algorithm}[1]{Algorithm~\ref{alg:#1}}
    \newcommand{\AppendixName}[1]{\label{app:#1}}
    \newcommand{\Appendix}[1]{Appendix~\ref{app:#1}}
    \newcommand{\ClaimName}[1]{\label{clm:#1}}
    \newcommand{\Claim}[1]{Claim~\ref{clm:#1}}
    \newcommand{\CorollaryName}[1]{\label{cor:#1}}
    \newcommand{\Corollary}[1]{Corollary~\ref{cor:#1}}
    \newcommand{\Fact}[1]{Fact~\ref{fact:#1}}
    \newcommand{\FactName}[1]{\label{fact:#1}}
    \newcommand{\Figure}[1]{Figure~\ref{fig:#1}}
    \newcommand{\FigureName}[1]{\label{fig:#1}}
    \newcommand{\LemmaName}[1]{\label{lem:#1}}
    \newcommand{\Lemma}[1]{Lemma~\ref{lem:#1}}
    \newcommand{\Section}[1]{Section~\ref{sec:#1}}
    \newcommand{\SectionName}[1]{\label{sec:#1}}
	\newcommand{\Theorem}[1]{Theorem~\ref{thm:#1}}
    \newcommand{\TheoremName}[1]{\label{thm:#1}}
    \renewcommand{\comment}[1]{}

\usepackage[textsize=scriptsize]{todonotes}

\newcommand{\eps}{\varepsilon}
\newcommand{\E}{\mathbb{E}}
\newcommand{\R}{\mathbb{R}}
\newcommand{\cA}{\mathcal{A}}
\newcommand{\cB}{\mathcal{B}}
\newcommand{\cC}{\mathcal{C}}
\newcommand{\cS}{\mathcal{S}}
\newcommand{\cX}{\mathcal{X}}
\DeclareMathOperator{\OPT}{OPT}
\newcommand{\mem}{\mu}
\newcommand{\wmax}{w_\text{max}}
\newcommand{\wmin}{w_\text{min}}
\newcommand{\polylog}{\operatorname{polylog}}
\newcommand{\poly}{\operatorname{poly}}

\begin{document}

\copyrightyear{2018} 
\acmYear{2018} 
\setcopyright{acmlicensed}
\acmConference[SPAA '18]{30th ACM Symposium on Parallelism in Algorithms and Architectures}{July 16--18, 2018}{Vienna, Austria}
\acmBooktitle{SPAA '18: 30th ACM Symposium on Parallelism in Algorithms and Architectures, July 16--18, 2018, Vienna, Austria}
\acmPrice{15.00}
\acmDOI{10.1145/3210377.3210386}
\acmISBN{978-1-4503-5799-9/18/07}

\settopmatter{printacmref=false}

\title{Greedy and Local Ratio Algorithms in the MapReduce Model}
\author{Nicholas J. A. Harvey}
\affiliation{
  \institution{University of British Columbia}
  \city{Vancouver}
  \country{Canada}
}
\email{nickhar@cs.ubc.ca}

\author{Christopher Liaw}
\affiliation{
  \institution{University of British Columbia}
  \city{Vancouver}
  \country{Canada}
}
\email{cvliaw@cs.ubc.ca}

\author{Paul Liu}
\affiliation{
  \institution{Stanford University}
  \city{Stanford}
  \state{California}
  \country{USA}
}
\email{paulliu@stanford.edu}

\begin{abstract}
MapReduce has become the \textit{de facto} standard model for
designing distributed algorithms to process big data on a cluster.
There has been considerable research on designing efficient
MapReduce algorithms for clustering, graph optimization, 
and submodular optimization problems.
We develop new techniques for designing greedy and local ratio
algorithms in this setting.
Our randomized local ratio technique gives
$2$-approximations for weighted vertex cover and weighted matching, and an $f$-approximation for weighted set cover, all in a constant number of MapReduce rounds.
Our randomized greedy technique gives algorithms for maximal independent set, maximal clique,
and a $(1+\eps)\ln \Delta$-approximation for weighted set cover.
We also give greedy algorithms for vertex colouring with $(1+o(1))\Delta$ colours and edge colouring with $(1+o(1))\Delta$ colours.
\comment{
Using variations of the local ratio theorem, we present two simple randomized approximation algorithms in the MapReduce model of Karloff et al \cite{KSV10}. The first is a 2-approximation for the minimum weight vertex cover, and the second is a 2-approximation for the maximum weight matching. Both algorithms run in a constant number of MapReduce rounds and improve significantly in approximation ratio, while preserving the low memory of existing algorithms. To our knowledge, our algorithm is the first for minimum weight vertex cover that runs in a constant number of rounds in the MapReduce model.
}
\end{abstract}



\maketitle


\section{Introduction}
\SectionName{intro}

A wealth of algorithmic challenges arise in processing the large data sets common in social networks, machine learning, and data mining.
The tasks to be performed on these data sets commonly involve graph optimization, clustering, selecting representative subsets, etc.\ --- tasks whose theory is well-understood in a sequential model of computation. 

The modern engineering reality is that these large data sets are distributed across clusters or data centers, for reasons including  bandwidth and memory limitations, and fault tolerance.
New programming models and infrastructure, such as MapReduce and Spark, have been developed to process this data 
efficiently.
The MapReduce model is attractive to theoreticians
as it is clean, simple, and has rigorous foundations in the work of Valiant \cite{V90} and Karloff et al.~\cite{KSV10}.
Designing MapReduce algorithms often involves concepts arising in parallel, distributed, and streaming algorithms, so
it is necessary to understand classic optimization problems in a new light.

\comment{
\vspace{3cm}

In recent years, the necessity of algorithms for large-scale data processing has grown rapidly, and several computational models for massive parallelization have been proposed. Such models build upon the familiar PRAM model, and restrict computation in realistic ways encountered in practice. MapReduce, as well as its open-source alternative, Hadoop, have gained widespread popularity following enthusiastic adoptation in industry for highly-parallel data analysis.}

In the formalization of the MapReduce model \cite{KSV10}, the data for a given problem is partitioned across many machines, where each machine has memory sublinear in the input size.
Computation proceeds in a sequence of rounds: in each round, a machine performs a polynomial-time computation on its local data.
Between rounds, each machine simultaneously sends data to all other machines, the size of which is restricted only by the sender's and recipients' memory capacity.
The primary efficiency consideration in this model is the number of rounds, although it is 
also desirable to constrain other metrics, such as the memory overhead, the amount of data communicated, and the processing time.

Greedy approximation algorithms have been 
a popular choice for adapting to the MapReduce model, in the hopes that their simple structure suits the restrictions of the model.
Unfortunately, many greedy algorithms also seem to be inherently sequential, a property which is rather incompatible with the parallel nature of MapReduce computations. 
This phenomenon is, of course, well known to parallel and distributed algorithm researchers, as it was the impetus for Valiant's \cite{V82} and Cook's \cite{C83} encouragement to study the maximal independent set problem in a parallel setting.
Primal-dual approximation algorithms have received comparably less use in the MapReduce setting, perhaps because they seem even more sequential than greedy algorithms.

\subsection{Techniques and Contributions}

We develop two new techniques for designing MapReduce algorithms. 
The first is a ``randomized local ratio'' technique, which combines the 
local ratio method \cite{BBFR04} with judicious random sampling in order to enable parallelization.
At a high level, the technique choses an i.i.d.\ random sample of elements, then runs an ordinary local ratio algorithm on this sample.
The crux is to show that the weight adjustments performed by the local ratio algorithm cause a significant fraction of the non-sampled elements to be eliminated.
Repeating this several times allows us to execute the local ratio method with no loss in the approximation ratio.
We use this technique to give a $2$-approximation for min-weight vertex cover (\Section{fSC}) and a $2$-approximation for max-weight matching (\Section{matching}).
The vertex cover result is a special case of our $f$-approximation for min-weight set cover, where $f$ is the largest frequency of any element.
Additional subtleties arise in applying this technique to the max-weight $b$-matching problem ($b \geq 2$), for which we obtain a $(3\!-\!\frac{2}{b}\!+\!\eps)$-approximation (\Appendix{bmatching}).
Adapting these randomized local ratio algorithms to the MapReduce setting is straightforward: all machines participate in the random sampling, and a single central machine executes the local ratio algorithm.
It is worth emphasizing that these algorithms are very simple and could easily be implemented in practice.



Our second technique we call the ``hungry-greedy'' technique.
Whereas randomized local ratio uses i.i.d.\ sampling, hungry-greedy first samples ``heavy'' elements, not to maximize an objective function, but to disqualify a large fraction of elements from
the solution.
Doing so allows us to rapidly shrink the problem size, and thereby execute the greedy method in just a few rounds.
We emphasize that the hungry-greedy technique is applicable even to problems without an objective function or an \textit{a priori} ordering on the greedy choices.
We use this technique to give efficient algorithms for
maximal independent set (\Section{MIS}),
maximal clique (\Appendix{maxclique}),
and a $(1+\eps)\ln \Delta$ approximation to weighted set cover (\Section{logDeltaSC}).

At first glance one may think that the maximal clique result follows by a trivial reduction from the independent set result, but this is not the case. 
The formalization of the MapReduce model requires that computations be performed in linear space, so it is not possible to complement a sparse graph. 
Other problems that are widely studied in the distributed computing literature and can usually be solved by a trivial reduction to maximal independent set include vertex colouring with $\Delta+1$ colours and edge colouring with $2\Delta-1$ colours; see the monograph of Barenboim and Elkin \cite{BE17}.
Again, these reductions are not possible in the MapReduce model due to space restrictions.
We make progress on these problems by giving algorithms for $(1+o(1))\Delta$ vertex colouring and $(1+o(1))\Delta$ edge colouring (\Section{colouring}).

\begin{figure*}[ht]
\onecolumn
\caption{\footnotesize Results for MapReduce algorithms.
For graph algorithms, the number of vertices is $n$, the number of edges is assumed to be $n^{1+c}$, and the maximum degree is $\Delta$.
For set cover, the number of sets is $n$,
the size of the ground set is $m$, 
the size of the largest set is $\Delta$,
and the largest number of sets containing any element is $f$, the maximum (resp.~minimum) weight is $\wmax$ (resp.~$\wmin$).
For \Theorem{fSCMR} it is assumed that $m = n^{1+c}$. The space per machine is typically $n^{1+\mem}$, which for most results (except \cite{AG15}) need not be constant; taking $\mem=1/\log n$ one can often get $O(n)$ space and $O(\log n)$ rounds.
For \Theorem{logDeltaSCMR}, $n$ can depend arbitrarily on $m$.
}
\FigureName{table}
\sffamily
\scriptsize
\setlength\extrarowheight{3pt}
\begin{longtabu} to \textwidth{X[1.9,l]|X[1,c]|X[1.5,l]|X[3.2,l]|X[1.6,l]|X[1,l]}
\textbf{Problem} & \textbf{Weighted?} & \textbf{Approximation} & \textbf{MapReduce Rounds} & \textbf{Space per machine} & \textbf{Reference}
\\[1pt]\hline
Vertex Cover
	& 
    & $2$
    & $O(c/\mem)$
    & $O(n^{1+\mem})$
    & \cite{KMVV15}
    \\
    & Y
    & $O(\log^2 n)$
    & $2$
    & $\tO(n^{1.5})$
    & \cite{AK17}
    \\
    & Y
    & $2$
    & $O(c/\mem)$
    & $O(n^{1+\mem})$
    & \Theorem{fSCMR}
    \\

    \hline
Set Cover
	&
    & $O(\log m)$
    & $O(\log(m) \card{OPT})$
    & not analyzed
    & \cite{MKBK15}
	\\
	&
	& $(1+\eps)\ln m$
	& $O\big(\log(\frac{n}{\card{OPT}}) \log(\Delta) / \eps + \log m\big)$
    & $O(\sqrt{n\card{OPT}}\Delta)$
    & \cite{MZK16}
    \\
    & Y
    & $f$
    & $O((c/\mem)^2)$
    & $O(f \cdot n^{1+\mem})$
    & \Theorem{fSCMR}
   	\\
    & Y
    & $(1+\eps)\ln \Delta$
    & $O\Bigg(\frac{\log\big(\frac{\wmax}{\wmin}\Delta\big)}{\mem^2 \eps}\Bigg)$ if $n=\operatorname{poly}(m)$
    & $O(m^{1+\mem})$
    & \Theorem{logDeltaSCMR}

	\\\hline
Maximal Indep.\ Set
	& 
    & 
    & $O(c/\mem)$
    & $O(n^{1+\mem})$
    & \Theorem{MIS2}
	\\

	\hline
Maximal Clique
	& 
    & 
    & $O(1/\mem)$
    & $O(n^{1+\mem})$
    & \Corollary{maxclique}
	\\


    \hline
Matching
	&
	& $2$
    & $O(c/\mem)$
    & $O(n^{1+\mem})$
    & \cite{KMVV15}
    \\
	& Y
	& $8$
    & $O(c/\mem)$
    & $O(n^{1+\mem})$
    & \cite{KMVV15}
    \\
	& Y
	& $4+\eps$
    & $O(c/\mem)$
    & $O(n^{1+\mem})$
    & \cite{CS14}
    \\
	& Y
	& $3.5+\eps$
    & $O(c/\mem)$
    & $O(n^{1+\mem})$
    & \cite{GMZ16}
    \\
	& Y
	& $2+\eps$
    & $O( (c/\mem) \log(1/\eps))$
    & $O(n^{1+\mem})$
    & \cite{LPP15}
    \\
	& Y
	& $1+\eps$
    & $O(1/\mem\eps)$
    & $O(n^{1+\mem})$
    & \cite{AG15}
    \\
    & Y
    & $O(1)$
    & $2$
    & $\tO(n^{1.5})$
    & \cite{AK17}
	\\
    & 
    & $2+\eps$
    & $O((\log \log n)^2 \log(1/\eps))$
    & $O(n)$
    & \cite{CLMMOS17}
    \\
    & 
    & $1+\eps$
    & $(1/\eps)^{O(1/\eps)} \log \log n$
    & $\tilde{O}(n)$
    & \cite{ABB17}
    \\
	& Y
	& $2$
    & $O(c/\mem)$
    & $O(n^{1+\mem})$
    & \Theorem{matchingMR}
	\\

    
    \hline
Vertex Colouring
	& 
    & $(1+o(1))\Delta$ colours
    & $O(1)$
    & $O(n^{1+\mem})$
    & \Theorem{vtxcolouring}
	\\[2pt]
\hline

Edge Colouring
	& 
    & $(1+o(1))\Delta$ colours
    & $O(1)$
    & $O(n^{1+\mem})$
    & \Theorem{edgecolouring}
	\\[2pt]
\hline
\end{longtabu}
\twocolumn
\end{figure*}



\subsection{Related work}

Within the theory community, part of the appeal of the MapReduce model is the connection to several established topics, including distributed algorithms, PRAM algorithms, streaming algorithms, submodular optimization and composable coresets. 

One of the earliest papers on MapReduce algorithms is due to Lattanzi et al.~\cite{LMSV11}, in which the filtering technique is introduced.
Filtering involves choosing a random sample, from which a partial solution is constructed and used to eliminate candidate elements from the final solution; the process is repeated until all elements are exhausted.
This technique has been quite influential, and many subsequent papers can be viewed as employing a similar methodology.
Indeed, our randomized local ratio technique can be viewed as a descendant of the filtering technique in which the local ratio weights are used to eliminate elements.
The original filtering technique yielded $2$-approximations for unweighted matching and vertex cover, and was combined with layering ideas to give an $8$-approximation for weighted matching \cite{LMSV11}.
These layering ideas were improved by Crouch and Stubbs \cite{CS14} and Grigorescu et al.~\cite{GMZ16} to give a $(3.5+\eps)$-approximation for weighted matching in the semi-streaming and MapReduce models.
In addition, it is also known that any $c$-approximation to maximum weighted matching can be automatically converted to a $(2+\eps)$-approximation algorithm
by using the former as a blockbox~\cite{LPP15}.

Recently Paz and Schwartzman~\cite{PS17} gave a $(2+\eps)$-approximation for weighted matching in the semi-streaming model, using the local ratio method with clever pruning techniques to bound the recursion depth.
This result was slightly improved and substantially simplified by Ghaffari~\cite{G17}.
Unfortunately, this work does not apply to the MapReduce setting: it is very space-efficient, but not distributed.
Our randomized local ratio technique is inspired by this work, but different.
Whereas \cite{PS17,G17} are deterministic and process the input in a streaming fashion, our technique achieves its space efficiency through random sampling and requires several rounds to process the entire input.

Submodular optimization has also been a 
fruitful avenue for MapReduce algorithms. 
Kumar et al.~\cite{KMVV15} gave a general framework for maximizing monotone submodular functions over hereditary constraints in the MapReduce setting while losing only $\eps$ in the approximation ratio.
Their Sample-and-Prune and $\eps$-Greedy techniques are vaguely related to our hungry-greedy technique, but nevertheless different as our technique does not maximize an objective.
Barbosa et al.~\cite{BENW16} develop a different framework that achieves the same approximation ratio while using fewer MapReduce rounds; it also supports the continuous greedy algorithm and
non-monotone functions.
A specific result relevant to our work is a $\frac{1}{3+\eps}$-approximation for weighted matching in bipartite graphs \cite{KMVV15} in $O(\log(\frac{\wmax}{\wmin})/\mem\eps)$ rounds and $O(n m^\mem)$ space per machine.
The results of Barbosa et al.~seem to improve this to $O(1/\eps)$ rounds and $O(\sqrt{nm})$ space per machine.

Another problem related to submodular maximization is the submodular set cover problem~\cite{W82}, which generalizes the weighted set cover problem, and has been studied\footnote{
	The set cover problem was studied earlier in the MapReduce setting by Stergiou and Tsioutsiouliklis~\cite{ST15}, although their work appears not to have rigorous guarantees regarding concurrent access to shared state.
}
in the MapReduce model by Mirzasoleiman et al.\ \cite{MKBK15,MZK16}.
Both of our techniques are applicable to weighted set cover.
Our randomized local ratio technique matches the classic $f$-approximation~\cite{BE85}, where $f$ is the largest frequency of an element.
Our hungry-greedy technique gives a $(1+\eps)\ln \Delta$ approximation,
nearly matching the optimal $\ln \Delta$-approximation for polynomial-time algorithms \cite{C79}.
It is worth noting that our $f$-approximation is intended (as with vertex cover) for the scenario that $n \ll m$, whereas our second algorithm assumes $m \ll n$.
Our latter result has several advantages over the work of Mirzasoleiman et al.: we handle the weighted case, improve the $\ln m$ in the approximation ratio to $\ln \Delta$, our space usage is controllable by an arbitrary parameter $\mem$, and use significantly less space when $m \ll n$.

The classical PRAM model
admits algorithms for many of the problems we consider.
There are general simulations of EREW PRAMs \cite{KSV10} and CREW PRAMs \cite{GSZ11} by MapReduce algorithms, although these lead to $\polylog(n)$-rounds algorithms, which does not meet the MapReduce gold standard of $O(1)$ rounds.

Composable core-sets have recently emerged as a useful primitive for the design of distributed approximation algorithms; see, e.g., \cite{BEL13,IMMM14,MZ15} and the references therein.
The idea is to partition the data, compute on each part a representative subset called a core-set, then solve the problem on the union of the core-sets.
Relevant to this is the work of Assadi et al.~\cite{ABB17} (extending work by Assadi and Khanna~\cite{AK17}), which designs $3/2 + \eps$ (resp.~$3+\eps$) approximate core-sets for unweighted matching (resp.~vertex cover).
They obtain 2-round MapReduce algorithms as a corollary. Combining with the technique of Crouch and Stubbs~\cite{CS14}, these can be extended to $O(1)$ (resp. $O(\log n)$) approximations for weighted variants of the problem.

Primal-dual algorithms have previously been studied in the MapReduce model, notably in the work of Ahn and Guha \cite{AG15}. 
They develop sophisticated multiplicative-weight-type techniques to approximately solve the matching (and $b$-matching) LP in very few iterations.
Their technique gives a $(1+\eps)$-approximation in $O(1/\mem\eps)$ MapReduce rounds while using only $O(n^{1+\mem})$ space per machine.
While their algorithm achieves a superior approximation factor to ours, it is also highly technical.
Our algorithm is concise, simple, and could plausibly be the end-of-the-road for filtering-type techniques for the matching problem.

Very recently there have been some exciting developments in MapReduce algorithms.
Im et al.~\cite{IMS17} have given a general dynamic programming framework that seems very promising for designing more approximation algorithms.
Czumaj et al.~\cite{CLMMOS17} introduced a round compression technique for MapReduce and applied it to matching to obtain a linear space MapReduce algorithm which gives a $(2+\eps)$-approximation algorithm for unweighted matching in only $O((\log \log n)^2 \log(1/\eps))$ rounds.
For unweighted matching, Assadi et al.~\cite{ABB17} built on this work and gave a $(2+\eps)$-approximation algorithm in $O((\log \log n) \log(1/\eps))$ rounds with $\tilde{O}(n)$ space and a $(1+\eps)$-approximation algorithm in $(1/\eps)^{O(1/\eps)} \log \log n$ rounds with $\tilde{O}(n)$ space.
Round compression has also been used to give a $O(\log n)$ approximation for unweighted vertex cover using $O(\log \log n)$ MapReduce rounds~\cite{Assadi17}.

\comment{
 - Emphasize importance of MIS in distributed literature

Done:
 - Interest to Data Mining, ML, Parallel algorithms and Theory Community.
 - Comparison to Lattanzi Filtering
 - Mention submodular literature (Kumar, Barbosa...)
 - Set cover:
    - n<<m and m>>n regimes
    - StergiouTsioutsiouliklis not rigorous
 - PRAM
 - PazSchwartzman
 - Composable Coresets
 - Cite Barbosa
 - Cite Moseley STOC17
}

\comment{
Though MapReduce is a relatively recent framework, there have already been numerous results adapting
classical algorithms to the MapReduce model \cite{KSV10,KMVV15,LMSV11,SV11}. In many cases, it has
been shown by Goodrich~\cite{G10} that $O(\log n)$ rounds of MapReduce are sufficient to simulate
certain types of PRAM algorithms, allowing simple translations between the two computational models.
For matching and vertex cover in particular, the method of \emph{filtering} by Lattanzi et
al.~\cite{LMSV11} provide simple, constant round randomized algorithms that achieve a constant
approximation ratio. For the problems of unweighted maximum matchings and unweighted minimum vertex covers,
they obtain a 2-approximation in $O(c/\mem)$ rounds, where the graph is assumed to have $O(n^{1+c})$
edges and each machine uses $O(n^{1+\mem})$ memory. 

To extend their maximal matching result to maximum \emph{weighted} matchings, they first partition the edge set of the graph into subsets of roughly equal weight.
Then, they run the maximal matching algorithm on each of the subsets and combine the resulting matchings into a single matching.
 They show that their algorithm achieves an 8-approximation in $O(c/\mem)$ MapReduce rounds.
Crouch and Stubbs~\cite{CS14} later gave a different parallelization scheme where each parallel instance consists only of edges with large enough weight.
Again, they compute a maximal matching on each parallel instance and then combine the matchings
greedily to obtain a $(4 + \eps)$-approximation.
Grigorescu, Monemizadeh, and Zhou~\cite{GMZ16} show that replacing the final greedy matching with an
optimal matching gives a $(3.5+\eps)$-approximation.
However, both \cite{CS14, GMZ16} could potentially replicate the dataset $O(\log n)$ times. In the case of weighted vertex covers, we are not aware of any approximations in the MapReduce model.
}

\comment{
\subsection{Contributions}
Using variations of the local ratio theorem \cite{BBFR04,BE85,PS17}, we present simple randomized
algorithms for minimum weight vertex cover and maximum weight matching. Our algorithms run in $O(c /
\mem)$ rounds, and beat all existing algorithms in terms of approximation ratio, and improving or matching them in number of rounds. Our approach is to sample the graph using the filtering method of Lattanzi et al.~\cite{LMSV11}, allowing a subset of the graph to fit on one machine. We then apply variants of the local ratio algorithm to the sampled graph, removing a large fraction of the edges in the input graph after each round of MapReduce. Our approximation for maximum weight matching matches the best known approximation for MapReduce maximum cardinality matching. To the best of our knowledge, our algorithm for minimum weight vertex cover is the first that runs in a constant number of rounds in the MapReduce model.
}

\subsection{The MapReduce Model}

In this paper we adopt the MapReduce model as formalized by Karloff et al.~\cite{KSV10}; see also \cite{G10}.
In their setting, there is an input of size $N$ distributed across $O(N^\delta)$ machines, each with $O(N^{1-\delta})$ memory. 
This models the property that modern big data sets  must be stored across a cluster of machines. 
Computation proceeds in a sequence of rounds.
In each round, each machine receives input of size $O(N^{1-\delta})$; it performs a computation on that input; then it produces a sequence of output messages, each of which is to be delivered to another machine at the start of the next round.
The total size of these output messages must also be $O(N^{1-\delta})$.
In between rounds, the output messages are delivered to their recipients as input.

The data format and computation performed is not completely arbitrary.
Data is stored as a sequence of key-value pairs, and the computation is performed by the eponymous \textit{map} and \textit{reduce} functions that operate on such sequences.
The requirements of these functions are not particularly relevant to our work, so we refer the interested reader to standard textbooks \cite{LRU14} or the work of Karloff et al.~\cite{KSV10}.

For graph algorithms, we will
assume that the space per machine is bounded as a function of $n$, the number of vertices, whereas the input size is $O(N)$ where $N=m$ is the number of edges. As in Lattanzi et al.~\cite{LMSV11}, we will typically assume that the space per machine is $O(n^{1+\mem})$ and the graph has $m=n^{1+c}$ edges, $c>\mem$.
This conforms to the requirements of Karloff et al.\ with $\delta = \frac{c-\mem}{1+c}$. 

The memory constraints of graph problems in the MapReduce model stem from the early work of Karloff et al.~\cite{KSV10} and Leskovec et al.~\cite{LKF05}. In these papers, graph problems with $n$ nodes and $n^{1+c}$ edges were examined. Specifically, Leskovec et al.~\cite{LKF05} found that real world graph problems were often not sparse, but had $n^{1+c}$ edges where $c$ varied between 0.08 and larger than 0.5. In practice, MapReduce algorithms also require $\Omega(n)$ memory per machine. One such example is the diameter estimation algorithm of Kang et al.~\cite{KTAFL11}, in which $O((n+m)/M)$ memory is required per machine, where $M$ is the number of machines in total.

\comment{
In this section we 
describe the MapReduce 
model. For more details, see \cite{G10,KSV10} and the references therein.

In the MapReduce (MRC) model, computations are carried out as a sequence of \emph{map}, \emph{shuffle}, and \emph{reduce} steps that operate with $(key, value)$ pairs as input and output. A single \emph{round} of MapReduce consists of a map, shuffle, and reduce in that order. Roughly speaking, the map and shuffle steps distribute the data to parallel machines, and the reduce step performs local computation with the data it's given. To implement a MapReduce algorithm, an algorithm designer implements the map and reduce functions:
\begin{itemize}
\item A mapper $\mem$ takes in a $(key, value)$ pair and outputs any number of $(key, value)$ pairs. 
\item A reducer $\rho$ takes in a single key and a list of values, and produces any number of $(key, value)$ pairs. 
\end{itemize}
The computation of both the mapper and reducer must depend only on its input arguments. The shuffle phase, which is system dependent, is responsible for distributing the output of the mappers to the reducers. Most importantly, the shuffle phase guarantees that all $(key, value)$ pairs with the same key will get their values sent to the same reducer. 

The attractiveness of MapReduce is mainly due to two advantages over traditional PRAM algorithms: (1) the scheduling -- often the most difficult part of a standard PRAM algorithm -- is left entirely to the system implementation of the shuffler, and (2) an algorithm designer only has to implement two relatively simple functions per round, the mapper and the reducer. 

To capture the real world restrictions of MapReduce, Karloff et al. \cite{KSV10} introduced constraints to the mapper, the reducer, and the total memory used in the shuffle phase. In our particular setting, any MapReduce algorithm must satisfy the following two properties, dependent on the input size $N$:
\begin{itemize}
\item Both the mapper and reducer must use at most $O(N^{1-\delta})$ memory.
\item The total amount of data output by the mapper is at most $O(N^{2-2\delta})$.
\end{itemize}

Since many PRAM algorithms take $O(\log n)$ rounds to simulate, we seek MapReduce algorithms where the number of rounds is a constant.}


\paragraph{Related models.}
There have been a series of refinements to the original MRC model of Karloff et al.~\cite{KSV10}, the most prominent of which is the massive parallel computation (MPC) model developed by Beame et al.~\cite{BKS13}. The main difference with the MPC model is that the space requirements are more stringent than in MRC. Given an input of size $N$ and $M$ machines, each machine is only allowed to have at most $S = O(N/M)$ space. In each round, each machine can send $O(S)$ words to other machines. The majority of our algorithms apply in the MPC model as well as the MRC model. The only exceptions are our set cover and $b$-matching algorithms, where our space bound depends on structural parameters of the input.

\comment{
\subsection{The local ratio method}
The local ratio method was introduced by Bar-Yehuda and Even as a method for finding a low cost vertex cover of a graph \cite{BE85}.
Since then, it has found numerous applications; see \cite{BBFR04} for an excellent survey.

The intuition behind the local ratio method is as follows.
Suppose we wish to find an $x \in \cC \subseteq \R^n$ that minimizes the quantity $w^{\top} x$ where $w \in \R^n$ is some fixed vector.
In the local ratio method, we try to find a decomposition of $w$, say $w = w_1 + w_2$, such that both $w_1^\top x$ and $w_2^\top x$ can be approximately minimized by a single $x$. That is, assuming that $x_1^*$ and $x_2^*$ are the minima to $w_1$ and $w_2$ respectively, we find a single $x$ for which $w_i^\top x \leq \alpha w_i^\top x_i^*$, where $\alpha \geq 1$ is the approximation ratio. It is often the case that $w_2$ is chosen based on structural properties of the problem, exploiting locality in some fashion. For example, in our local ratio algorithm for weighted vertex cover, the weight function $w_2$ is determined by choosing an arbitrary edge from the input graph and making $w_2$ non-zero only on the endpoints of the edge. Moreover, a desirable property is that a large fraction of weight function $w_1$ can be ignored, while maintaining good approximation properties on $w_2$. Repeating the decomposition on $w_1$ allows us to eventually reach the trivial case $w_1 = 0$ in a small number of iterations. In our applications, our choice of $w_2$ is drawn i.i.d. from a class of possible decompositions.

We now give two examples of this framework which form the basis of our MapReduce algorithms for minimum vertex cover and maximum matching. Sequential and streaming variants of these local ratio algorithms appear in \cite{BBFR04,PS17}. However these algorithms cannot be used as is, due to the restrictive memory and round complexity requirements of the MapReduce model. Although there are fairly standard alternative proofs of these algorithms, we give a unified approach using LP duality for the sake of presentation.
}

\subsection{Organization}
In Sections~\ref{sec:fSC} and \ref{sec:MIS}, we develop our ``randomized local ratio'' and ``hungry-greedy'' techniques through the weighted set cover and maximal independent set problems respectively. We then apply these techniques to specific problems in the subsequent sections. Our randomized local ratio technique is applied to maximum weight matching (\Section{matching}) and $b$-matching (\Appendix{bmatching}), and our hungry-greedy technique is applied to maximal clique (\Appendix{maxclique}) as well as an alternate approximation of weighted set cover (\Section{logDeltaSC}). Algorithms for vertex and edge colouring that may be of independent interest are given in \Section{colouring}.

\section{$f$-approximation for weighted set cover}
\SectionName{fSC}

For notational convenience, let $[n] = \set{1,\ldots,n}$, $S(X) = \bigcup_{i \in X} S_i$
and $w(X) = \sum_{i \in X} w_i$ for all $X \subseteq [n]$.
%
In the weighted set cover problem, we are given $n$ sets $S_1, \ldots, S_n \subseteq [m]$ with weights $w_1, \ldots, w_n \in \R_{> 0}$.
The goal is to find a subset $X$ such that
$w(X)$
is minimal amongst all $X$ with $S(X) = [m]$.

Let $OPT$ denote a fixed set $X$ achieving the minimum.
The frequency of element $j \in U$ is defined to be $\card{ \setst{ i }{ j \in S_i } }$.
Let $f$ denote the maximum frequency of any element.

\begin{theorem}[Bar-Yehuda and Even~\cite{BE85}]
	\label{thm:lr_sc}
	The following algorithm gives an $f$-approximation to weighted set cover.
\end{theorem}

\begin{tcolorbox}[title=Sequential local ratio algorithm for minimum weight set cover]
	Arbitrarily select an element $j \in U$ such that the minimum weight of all sets containing
    $j$ is strictly positive, then reduce all those weights by that minimum value.
    Remove all sets with zero weight, and add them to the cover.
    Repeat so long as uncovered elements remain.
\end{tcolorbox}

More explicitly, the weight reduction step works as follows.
If element $j \in U$ was selected, then we
compute $\epsilon = \min_{i \in [n] \::\: j \in S_i} w_i$
then perform the update $w_i \gets w_i - \epsilon$
for all $i$ with $j \in S_i$.

\comment{
\begin{proof}
	The output of the algorithm gives a valid set cover.
    To see this, suppose that an element was not covered.
    Then all sets containing this element must have positive weight
    (since all sets with zero weight are part of the cover).
    This is impossible, as the algorithm does not terminate so long as there are elements 
    for which all sets containing it have positive weight.

    \textbf{TODO}
    \comment{
	We now prove that the set cover that is output by the algorithm is indeed an $f$-approximation.
	The linear program relaxation for minimum weight set cover is
    \begin{equation}
      \label{eqn:vc_p}
      \min \sum_{v \in V} w_v x_v \quad \text{s.t.} \quad x_u + x_v \geq 1 \,\, \forall (u,v) \in E, \quad x \geq 0.
   \end{equation}
   Taking the dual of \eqref{eqn:vc_p} gives
   \begin{equation}
      \label{eqn:vc_d}
      \max \sum_{e \in E} y_e \quad \text{s.t.} \quad \sum_{e \ni v} y_e \leq w_v \,\, \forall v \in V, \quad y \geq 0.
   \end{equation}
   For the proof, we can modify the basic algorithm to construct a feasible dual solution as follows.
   We initialize $y$ to be the $0$ vector.
   Consider the $t$-th iteration of the algorithm and suppose that we select the edge $e = (u,v)$.
   Let $w'_u, w'_v$ denote the current weights of the vertices $u, v$, respectively.
   Then we will increment $y_e$ by $\min\{w'_u, w'_v\}$ and decrement both $w'_u, w'_v$ by the same amount
   ($y_e'$ does not change for $e' \neq e$).
   To see that $y$ is feasible, we claim that at any iteration $t$ we have $w'_v + \sum_{e \ni v} y_e = w_v$.
   Indeed, this is true because whenever we select an edge $e$ incident to $v$, increment $w'_v$ and $y_e$ by the same amount.
   Finally, let $U$ be the set of vertices with weight $0$ at the end of the algorithm.
   Then $U$ is a vertex cover, as mentioned above, and since $y$ is feasible we have
   \[
      2\cdot \OPT \geq 2 \sum_{e \in E} y_e = \sum_{v \in V} \sum_{e \ni v} y_e \geq \sum_{v \in U} \sum_{e \ni v} y_e = \sum_{v \in U} w_v.
   \]
   To conclude, we have that $U$ is a 2-approximate vertex cover to the minimum weight vertex cover problem.
   }
\end{proof}
}

\subsection{Randomized local ratio}
\SectionName{fSCRLR}


The local ratio method described above is not terribly parallelizable, as it processes elements sequentially.
Nor is it terribly space-efficient, as it might require processing every element in $[m]$. 
However, it does have the virtue that elements can be processed in a fairly arbitrary order.
Our randomized local ratio technique takes advantage of that flexibility and combines the local ratio algorithm with random sampling. There are essentially two components to our randomized local ratio technique: (1) we can cover a large fraction of the remaining sets with just a random sample of elements, and (2) this random sample will be comparable in weight to the optimal solution since the processing order of local ratio is fairly arbitrary. Intuitively, property (1) ensures that the algorithm runs in a few number of rounds, and property (2) ensures a good approximation ratio. These two ideas will come up frequently in the algorithms we present in this paper. 

In \Algorithm{minsc}, the size of each sample $U'$ is $O(\eta)$ w.h.p., and $\eta$ will be taken to be $n^{1+\mem}$, the space available on each machine.

\begin{algorithm}
    \caption{\small
    	An $f$-approximation for minimum weight set cover.
    	The lines highlighted in blue are run sequentially on a central machine, 
    	and all other lines are run in parallel across all machines.
    }
    \label{alg:minsc}
    \begin{algorithmic}[1] 
        \Procedure{ApproxSC}{$S_1,\ldots,S_n \subseteq U = [m]$}
        	\State $U_1 \gets U$, $r \gets 1$
            \While{$U_r \neq \emptyset$}
                \State $\rhd$: $U_r$ is the set of all $j \in [m]$ such that $\min_{i \::\: j \in S_i} w_i > 0$
            	\State \label{alg:minscsample}
                    Compute $U'$ by sampling each element of $U_r$
                    independently with probability $p = \min(1, \frac{2\eta}{|U_r|})$ 
                 \State \textbf{if $|U'| > 6\eta$ then fail}
				\CSTATET Run the local ratio algorithm for set cover on the sets
                    $S_1 \intersect U', \ldots, S_n \intersect U'$
                    \label{alg:minsclocalratio}
                \CSTATE Let $C \subseteq [n]$ be the indices of all sets with zero weight
                	\label{alg:minscC}
                \State $U_{r+1} \gets U_r \setminus S(C)$
                    \label{alg:minscresidual}
                \State $r \gets r + 1$
            \EndWhile
            \State \textbf{return} the indices $C$ of all sets with zero weight 
        \EndProcedure
    \end{algorithmic}
\end{algorithm}

Our analysis uses the following useful fact, relating to the filtering technique of Lattanzi et al.~\cite{LMSV11}.

\begin{lemma}
	\LemmaName{filterSC}
	Let $U_{r+1}$ be the set computed on line \Cref{alg:minscresidual} of \Cref{alg:minsc}.
    If $p=1$ then $U_{r+1}=\emptyset$,
    otherwise $\card{U_{r+1}} < 2n/p$ with probability at least $1 - e^{-n}$.
\end{lemma}
\begin{proof}
    If $p=1$ then $U'=U_r$,
    so the local ratio method produces $C$ covering all of $U_r$,
    and so $U_{r+1} = \emptyset$. So assume $p<1$.

    For each $X \subseteq [n]$,
    let $\overline{S(X)}$ be the elements left uncovered by $X$.
    Let $\cE_X$ be the event that $U' \intersect \overline{S(X)} = \emptyset$.
    Say that $X$ is large if $\card{\overline{S(X)}} \geq 2n/p =: \tau$.
    For large $X$, $\Pr[\cE_X] \leq (1-p)^\tau \leq e^{-p \tau} = e^{-2n}$.
    By a union bound, the probability that all large $\cE_X$ fail to occur is at least $1-e^{-n}$.
    Since the local ratio method produces a set $C$ for which
    $U' \intersect \overline{S(C)} = \emptyset$,
    with high probability it must be that $|C|$ is not large.
\end{proof}

\begin{theorem}
    \TheoremName{SCf}
	Suppose that $m \leq n^{1+c}$.
    Let $\eta = n^{1+\mem}$, for any $\mem>0$.
    Then \Cref{alg:minsc} terminates within $\ceil{c / \mem}$ iterations
    and returns an $f$-approximate set cover, w.h.p.
\end{theorem}
\begin{proof}
    \textit{Correctness:}
        Observe that the sequential local ratio algorithm for set cover
        can pick elements in an arbitrary order.
        \Algorithm{minsc} is an instantiation of that algorithm
        that uses random sampling to partially determine the order in which elements are picked.
        It immediately follows that \Algorithm{minsc} outputs an $f$-approximation to
        the minimum set cover, assuming it does not fail.
        By a standard Chernoff bound, an iteration fails with probability at most $\exp(-n^{1+\mem}) \leq \exp(-n)$.

	\textit{Efficiency:}
    By \Lemma{filterSC}, while $p<1$ we have
    $|U_{r+1}| \leq 2n/p = 2 n |U_r| / \eta = 2 |U_r| / n^{\mem} \leq 2 n^{1 + c - r \mem}$, whp.
    By iteration $r=\ceil{ c / \mem }-1$, we will have $\card{U_r} \leq 2 n^{1+c} = 2 \eta$ and so $p=1$.
    After one more iteration the algorithm will terminate, again by \Lemma{filterSC}.
\end{proof}

\subsection{MapReduce implementation}

\begin{theorem}
\TheoremName{fSCMR}
There is a MapReduce algorithm that computes an $f$-approximation for the minimum weight set cover
where each machine has memory $O(f \cdot n^{1+\mem})$.
The number of rounds is $O(c/\mem)$ in the case $f=2$ (i.e., vertex cover)
and $O((c/\mem)^2)$ in the case $f>2$.
\end{theorem}
\begin{proof}
Instead of representing the input as the sets $S_1,\ldots,S_n$, we instead assume that
it is represented as the ``dual'' set system $T_1,\ldots,T_m$ where $T_j = \setst{ i }{ j \in S_i }$.
By definition of frequency, we have $\card{T_j} \leq f$.

Each element $j \in [m]$ will be assigned arbitrarily to one of the machines, with $n^{1+\mem}$ elements per machine, so $M:=m/n^{1+\mem}=n^{c-\mem}$ machines are required.
Element $j$ will store on its machine the set $T_j$ and a bit indicating if $j \in U_r$.
Thus, the space per machine is $O(f \cdot n^{1+\mem})$.

The main centralized step
is the local ratio computation (lines \ref{alg:minsclocalratio}-\ref{alg:minscC}).
The centralized machine receives as input the sets $T_j$ for all $j \in U'$.
As $\card{U'} \leq 6 \eta$, the input received by the central machine is $O(f \eta) = O(f n^{1+\mem})$ words.
After executing the local ratio algorithm, the centralized machine then computes $C$.

There are two key steps requiring parallel computation:
the sampling step (line \ref{alg:minscsample})
and the computation of $U_{r+1}$ (line \ref{alg:minscresidual}).
After the central machine computes $C$, it must send $C$ to all other machines.
Sending this directly could require $\card{C}\cdot M = \Omega(n^{1+c-\mem})$ bits of output,
which could exceed the machine's space of $O(n^{1+\mem})$.
Instead, we can form a broadcast tree over all machines with degree $n^{\mem}$ and depth $c/\mem$, allowing us to send $C$ to all machines in $O(c/\mem)$ MapReduce rounds.
Since each machine knows $C$, each element $j$ can determine if $j \in U_{r+1}$ by determining if $T_j \intersect C \neq \emptyset$.
The same broadcast tree can be used to compute $\card{U_{r+1}}$ and send that to all machines.
Since each element $j$ knows if $j \in U_{r}$ and knows $\card{U_{r}}$, it can independently perform the sampling step.

In the case $f=2$ (i.e., vertex cover), this can be improved. 
Each set $S_i$ (i.e., vertex) will also be assigned to one of the $M$ machines, randomly chosen.
By a Chernoff bound (using that $\card{S_i}\leq n$) the space required per machine is $O(f \cdot n^{1+\mem})$ w.h.p.
After computing $C$, the central machine sends a bit to each set $S_i$ indicating whether $i \in C$ or not.
Each set $S_i$ then forwards that bit to each element $j \in S_i$.
Thus $j$ can determine whether $j \in U_{r+1}$.
Each machine can send to a central location the number of edges on that machine that lie in $U_{r+1}$, so the central machine can compute $\card{U_{r+1}}$ and send it back to all machines.
%
\end{proof}

\section{Maximal independent set}
\SectionName{MIS}

As a warm-up to the ``hungry-greedy'' technique, we first present 
an algorithm to  find a maximal independent set (MIS) in a constant number of MapReduce rounds.
In order to clearly explain the concepts, this section presents a simple algorithm using $O(1/\mem^2)$ rounds. 
In \Appendix{improvedMIS}, we show how to further parallelize this algorithm so that it takes $O(c/\mem)$ rounds.

The algorithm, shown in \Algorithm{mis}, proceeds in phases where each phase takes $O(1/\mem)$ rounds. In phase $i \geq 1$, we reduce the maximum degree from $n^{1 - (i-1)\alpha}$ to $n^{1 - i\alpha}$.
We will choose $\alpha = \mem/2$ so that after $O(1/\mem)$ phases, the maximum degree will be at most $n^\mem$. Following this, we can finish the algorithm in one more round by placing the entire graph onto a central machine.

For a vertex $v$, we define $N(v)$ as the neighbours of $v$ and $N^+(v) = N(v) \cup \{v\}$.
For a set $I \subseteq V$, we define $N^+(I) = \cup_{v \in I} N^+(v)$ and $N_I(v) = N(v) \setminus N^+(I)$.
In other words, $N_I(v)$ is the set of neighbours of $v$ which are not in $I$ and not adjacent to a vertex in $I$.
Finally let $d_I(v) = \card{N_I(v)}$ if 
$v \notin N^+(I)$, otherwise $d_I(v)=0$.

\begin{algorithm}
    \caption{\small 
        A simple algorithm for maximal independent set.
    	The lines highlighted in blue are run sequentially on a central machine, 
    	and all other lines are run in parallel across all machines.
    }
    \label{alg:mis}
    \begin{algorithmic}[1] 
        \Procedure{MIS1}{$G = (V,E)$}
        	\State $I \gets \emptyset$ \Comment{$I$ is the current independent set}
            \For{$i = 1, \ldots, 1/\alpha$}
               \State $\rhd$: $d_I(v) \leq n^{1-(i-1)\alpha}$ for all $v \in V$
               \State $V_H \gets \{v : d_I(v) \geq n^{1-i\alpha}\}$ \Comment{Heavy vertices} \label{ln:mis_vh1}
            	\While{$|V_H| \geq n^{i\alpha}$}
                  \State Draw $n^{i\alpha}$ groups of $n^{\mem/2}$ vertices from $V_H$, say $\cX_1, \ldots, \cX_{n^{i\alpha}}$ \label{ln:mis_vertex_sampling}
                  \CFOR{$j = 1, \ldots, n^{i\alpha}$}
                     \CIF{$\exists\, v_{j} \in \cX_{j}$ such that $d_I(v_{j}) \geq n^{1 - i\alpha}$}
	                    	\CSTATE $I \gets I \cup \{v_j\}$
	                  \ENDIF
	               \ENDFOR
                  \State $V_H \gets \{v : d_I(v) \geq n^{1-i\alpha}\}$ \Comment{Update heavy vertices} \label{ln:mis_vh2}
	            \EndWhile
            \CSTATE Find maximal independent set in $V_H$ and add it to $I$ \Comment{$|V_H| < n^{i\alpha}$}
            \EndFor
            \State \textbf{return} $I$
        \EndProcedure
    \end{algorithmic}
\end{algorithm}

\begin{remark}
    In the algorithm, a vertex can 
    lose its label as a heavy vertex during the for loop.
    This is intentional as when we add a vertex to $I$, we want to make sure we make substantial progress.
    Moreover, since a vertex sends all its neighbours, it is easy to tell whether a vertex is heavy and to update $N^+(I)$.
\end{remark}
\begin{lemma}
   \label{lem:mis_progress1}
   Let $V_H$ be the set of heavy vertices at the beginning of the inner for loop (line 9) 
   and $V_H'$ be the set of heavy vertices after the for loop.
   Then $|V_H'| \leq |V_H| / n^{\mem/4}$ w.h.p.
\end{lemma}
\begin{proof}
	 Observe that the first group $\cX_1$ of vertices will definitely contain a heavy vertex
    so it will remove at least $n^{1 - i\alpha}$ vertices from the graph.\footnote{We will say $v$ is removed from the graph if it is added to $N^+(I)$.}
    Now suppose we are currently processing group $\cX_j$.
    If the number of heavy vertices is at most $|V_H| / n^{\mem/4}$ at this point then we are done. So suppose the number of heavy vertices is at least $|V_H| / n^{\mem/4}$.
    Since we sample the vertices uniformly at random,
    the group $\cX_j$ contains a heavy vertex with probability at least $1 - \left( 1 - 1/n^{\mem/4} \right)^{n^{\mem/2}} \geq 1 - \exp\left( -n^{\mem/4} \right)$.
    At this point, we add another heavy vertex to $I$ and remove it and all its neighbours.
    
    The above process can happen at most $n^{i\alpha}$ times before the number of heavy
    vertices is at most $|V_H|/n^{\mem/4}$.
    This is because if it happens $n^{i\alpha}$ times then 
    no vertices remain. 
    Hence, by taking a union bound over all groups, we have that $|V_H'| \leq |V_H| / n^{\mem/4}$ with probability at least $1 - n \cdot \exp\left( -n^{\mem/4} \right)$.
\end{proof}

\begin{theorem}
	There is a MapReduce algorithm to find a maximal independent set in $O(1/\mem^2)$ rounds and $O(n^{1+\mem})$ space per machine, w.h.p.
\end{theorem}

\begin{proof}[Proof sketch]
As before, there are $M := n^{c-\mem}$ machines.
Each vertex and its adjacency list is assigned to one of the $M$ machines, randomly chosen.
By a Chernoff bound, the space required per machine is $O(n^{1+\mem})$ whp.
Each vertex $v$ will maintain its value of $d_I(v)$, allowing it to determine if $v \in V_H$.
Thus, the sampling step (line \ref{ln:mis_vertex_sampling}) can be performed in parallel.

A central machine maintains and updates the sets $I$ and $N(I)$.
It receives as input the sets of vertices $\cX_1,\ldots,\cX_{n^{i \alpha}}$ and their lists of alive neighbours, so the total input size is proportional to
$$
\sum_{j=1}^{n^{i \alpha}} \, \sum_{v \in \cX_j} d_I(v)
~\leq~ n^{i \alpha} \cdot n^{\mem/2} \cdot n^{1-(i-1)\alpha}
~=~ n^{1+\mem/2+\alpha}.
$$
After executing the for-loop, the central machine can use these lists of neighbours to also compute $N^+(I)$.

The next step requiring parallelization is line \ref{ln:mis_vh2}.
To execute this, the central machine sends a bit to each vertex $v$ indicating if $v \in N^+(I)$.
Then, every $v \not\in N^+(I)$ asks each neighbour $w \in N(v)$ if $w \in N^+(I)$.
The results of these queries allow $v$ to compute $d_I(v)$.

By Lemma~\ref{lem:mis_progress1}, the while loop takes $O(1/\mem)$ rounds, so the entire algorithm takes $O(1/\mem^2)$ rounds.
\end{proof}

\paragraph{Maximal clique.}
One might initially assume that an algorithm for maximal independent set immediately implies an algorithm for maximal clique.
But, as mentioned in \Section{intro}, it is not clear how to reduce maximal clique to maximal independent set in $O(m)$ space:
complementing the graph might require $\Omega(n^2)$ space.
Nevertheless, it is possible to adapt \Algorithm{mis} so that it will compute a maximal clique.
A description of the necessary modifications appears in \Appendix{maxclique}.

\section{$(1+\eps) \ln \Delta$-approximation for weighted set cover}
\SectionName{logDeltaSC}

\comment{
TODO: Berger et al.~have a really simple way to preprocess the input so that $w_{\max} / w_{\min} \leq mn / \eps$ while increasing
the approximation ratio by only a factor of $1+\eps$.
We can probably mention this somewhere..
}

For convenience, we reuse the same notation from \Section{fSC}. The standard greedy algorithm for weighted set cover is as follows.
We maintain a set $C$ of \emph{covered} elements (initially $C = \emptyset$).
In each iteration, we find the set $S_i$ which maximizes $|S_i \setminus C| / w_i$.
We add $S_i$ to our solution and add the elements of $S_i$ to $C$.
It is known \cite{C83} that this algorithm has approximation ratio $H_{\Delta}$ where $\Delta = \max_\ell |S_\ell|$ and $H_{k} = \sum_{i=1}^k \frac{1}{k}$.

Unfortunately, implementing this greedy algorithm 
is tricky in MapReduce but, following Kumar et al.~\cite{KMVV15}, we can implement the \textit{$\eps$-greedy algorithm}, which differs from the standard greedy method as follows.
Instead of choosing the set $S_i$ to maximize $|S_i \setminus C| / w_i$,
we choose $S_i$ such that $|S_i \setminus C| / w_i \geq \frac{1}{1+\eps} \max_{j} \{ |S_j \setminus C| / w_j \}$.
The standard dual fitting argument \cite{C83,V01} can be easily modified to prove that this gives a $(1+\eps) H_{\Delta}$-approximate solution.

Our algorithm for set cover is also inspired by some of the PRAM algorithms for set cover (see, for example, \cite{BRS94,RV98,BPT11,DBS17}).
Indeed, we use a similar bucketing approach as \cite{BRS94,RV98,BPT11,DBS17} which we now describe.
Let $L = \max_{\ell}\{|S_{\ell}| / w_{\ell} \}$.
We first consider only the ``bucket'' of sets that have a cost ratio of at least $L / (1+\eps)$ and continue to add sets from this bucket to the cover until the bucket is empty (after which we decrease $L$ by $1+\eps$ and repeat).
However, note that once we add a set from the bucket to the cover, some of the sets in the bucket may no longer have a sufficiently cost ratio to remain in the bucket; in this case, we need to remove them from the bucket.

It is shown in \cite{BPT11} (improving on \cite{BRS94}) that exhausting a bucket can be done in $O(\log^2 n)$ time in the PRAM model; this can be easily translated to a $O(\log^2(n) / (\mem \log(m)))$-round algorithm in the MapReduce model.
Our main contribution is to show that, in MapReduce, one can exhaust a bucket in $O(\log^2(n) / (\mem^2 \log^2(m)))$ rounds.
In particular, when $n = \poly(m)$ and $\mem$ is a constant, the number of rounds decreases from $O(\log n)$ to $O(1)$.

\comment{
It was shown by \cite{BRS94}, and later improved by \cite{BPT11}, that exhausting a bucket can be done in $O(\polylog(n))$ time in the PRAM model.

Unfortunately, implementing this greedy algorithm in MapReduce (or the CRCW PRAM model) is tricky.
However, we can use a bucketing technique where instead of choosing the set $S_i$ to maximize $|S_i \setminus C| / w_i$,
we choose $S_i$ such that $|S_i \setminus C| / w_i \geq \frac{1}{1+\eps} \max_{j} \{ |S_j \setminus C| / w_j \}$.
(Such an approach has been used in, for example, \cite{BRS94,RV98,BPT11,DBS17,KMVV15}.)
The standard dual fitting argument \cite{C83,V01} can be easily modified to prove that this gives a $(1+\eps) H_{\Delta}$-approximate solution.

Using an $\eps$-approximate greedy choice in each iteration instead of the optimal greedy choice allows one to obtain a fast parallel algorithm.
}

\comment{
Let us first give a light exposition of the parallel algorithm for approximate minimum set cover.
We will use $C$ to denote the current elements 
in the partial solution constructed so far.
Initially, $C = \emptyset$.
The algorithm begins by setting $L = \max_\ell \{ |S_{\ell}| / w_{\ell} \}$. Let us say a set $S_\ell$ is \emph{profitable} if $|S_{\ell} \setminus C| / w_{\ell} \geq L/(1+\eps)$.
In each iteration, it partitions the profitable sets into $1/\alpha$ groups where group $i$ consists of sets whose cardinality is in $[m^{1-(i+1)\alpha}, m^{1-i\alpha})$.
It then randomly samples $m^{1-(i+1)\alpha}$ collections of $m^{\mem / 2}$ sets from group $i$.
Starting from group $1$, for each collection the algorithm then finds the set $S_{\ell}$ in the collection maximizing $|S_{\ell} \setminus C| / w_{\ell}$,
adds it to the current partial solution, and updates $C$.
The crucial observation is that (modulo a very small probability) if the number of profitable sets in group $i$ has not decreased significantly
then the chosen $S_{\ell}$ is profitable.
It will also turn out that we can decrease a certain potential function by a large amount in each iteration so within a few iterations, there will be no more profitable sets.
If we have not covered the universe at this point then we decrease $L$ by $(1+\eps)$ and repeat this process.
}

The high level idea of our parallel algorithm for approximate minimum set cover is as follows.
Let $C$ denote the current elements in the partial solution constructed so far (initially $C = \emptyset$).
At this point, the algorithm considers only considers sets, $S_{\ell}$, which are almost optimal, i.e.~$|S_{\ell} \setminus C| / w_{\ell} \geq L(1+\eps)$,
where initially, $L = \max_\ell \{ |S_{\ell}| / w_{\ell} \}$ but is decreased as the algorithm runs.
We then partition these sets into $1/\alpha$ groups where group $i$ consists of
sets whose cardinatliy is in $[m^{1-(i+1)\alpha}, m^{1-i\alpha})$.
From each group $i$, we then sample $m^{1-(i+1)\alpha}$ collections of $m^{\mem / 2}$ sets.
Starting from the group $1$, we will see that either every collection contains an almost optimal set or we have made a large amount of progress.
It will also turn out that we can decrease a certain potential function by a large amount in each iteration so within a few iterations, there will be no more profitable sets.
If we have not covered the universe at this point then we decrease $L$ by $(1+\eps)$ and repeat this process.

We now begin with a more formal treatment of the algorithm for minimum weight set cover.
We will assume that each machine has $O(m^{1+\mem} \log(n))$ space and $\sum_{i=1}^n |S_i| \geq m^{1+c}$ so that the memory is sublinear in the input size.
The pseudocode for the algorithm is given in Algorithm~\ref{alg:set_cover}.

\begin{algorithm}[t]
   \caption{\small
   		A $(1+\eps)\ln \Delta$-approximation for minimum weight set cover.
        Blue lines are centralized.}
   \label{alg:set_cover}
   \begin{algorithmic}[1] 
      \Procedure{ApproxSC}{$S_1,\ldots,S_n \subseteq U = [m], ~w_1,\ldots,w_n \in \bR_{>0}$}
         \State $L \gets \max_\ell \left\{ \frac{|S_\ell|}{w_\ell} \right\}$
         \State $\cS \gets \emptyset$
         \State $C \gets \emptyset$ \Comment{$C \subseteq [m]$ maintains set of \emph{covered} elements}
         \While{$C \neq [m]$}
            \State $k \gets 1$
            \While{$\exists \ell$ s.t.~$|S_\ell \setminus C_k|/w_\ell \geq L/(1+\eps)$} \label{ln:setalg_inner}
               \State $C_k \gets C$ \Comment{Maintain temporary set of covered elements; used only for proof.}
               \State Let $\cS_{k,i} = \{S_\ell : m^{1 - i\alpha} \leq |S_\ell \setminus C_k| < m^{1 - (i-1)\alpha} \text{ and } |S_{\ell} \setminus C_k| / w_k \geq L/(1+\eps) \}$
               \For{$i = 1, \ldots, 1/\alpha$}
                  \For{$j = 1, \ldots, 2m^{(i+1)\alpha}$}
                     \State Include each $S_\ell \in \cS_{k,i}$ into group $\cX_{i,j}$ with probability $\min\{1, m^{\mem/2} / |\cS_{k,i}|\}$
                  \EndFor
               \EndFor
               \If{$|\cX_{i,j}| \leq 4m^{\mem/2}$}
                  \State Send each $\cX_{i,j}$ to the central machine.
               \Else
                  \Comment{At least one $\cX_{i,j}$ is too big so fail and continue to next iteration.}
                  \State $k \gets k + 1$ \label{ln:alg_sc_fail}
                  \State \textbf{continue}
               \EndIf
               \CFOR{$i = 1, \ldots, 1/\alpha$}
                  \CFOR{$j = 1, \ldots, 2m^{(i+1)\alpha}$}
                     \CIFT{$\exists S_{\ell} \in \cX_{i,j}$ such that $|S_{\ell} \setminus C| \geq m^{1 - (i+1) \alpha} / 2$}
                        \CSTATE $\cS \gets \cS \cup \{S_{\ell}\}$
                        \CSTATE $C \gets C \cup S_{\ell}$
                     \ENDIF
                  \ENDFOR
               \ENDFOR
               \State $k \gets k + 1$
            \EndWhile
            \State $L \gets L / (1+\eps)$
         \EndWhile
         \State \textbf{return} $\cS$
      \EndProcedure
   \end{algorithmic}
\end{algorithm}

In \Section{logDeltaSCMR}, we will describe the MapReduce implementation of this algorithm.
Of particular note is that the second while loop can be implemented in a small number of MapReduce rounds.
Our goal in this section is to show that the number of times the while loop is executed is small.

To that end, let us fix $L$ and analyze the number of iterations of the second while loop starting at Line~\ref{ln:setalg_inner}.
To do this, it will be convenient to introduce the potential function
\[
   \Phi_k \coloneqq \sum_{\ell \in [n] : \frac{|S_\ell \setminus C_k|}{w_\ell} \geq \frac{L}{1+\eps} } |S_\ell \setminus C_k|.
\]
Observe that we have the trivial upper bound $\Phi_k \leq nm$.
Moreover, $\Phi_k = 0$ if and only if we finished with the second while loop by iteration $k$.

Using a straightforward Chernoff bound (Theorem~\ref{thm:chernoff}), we can show that each iteration has a small proabilility of failure,
i.e.~in each iteration, we make it into line~\ref{ln:alg_sc_fail} with very small probability.
\begin{claim}
   \label{claim:sc_failure}
   Fix an iteration $k$.
   We have $|\cX_{i,j}| \leq 4m^{\mem/2}$ for all $i, j$ with probability at least $1 - \frac{m}{\alpha} \exp\left( -m^{\mem/2} \right)$.
\end{claim}
The first lemma states that for a good fraction of the sets, a good fraction of their remaining elements has been covered.
This will be useful for showing that $\Phi_k$ decreases significantly after each iteration.

\begin{lemma}
   \label{lem:set_cover1}
   Fix $i$ and an iteration $k$.
   Let $\cS_{k,i}$ be as in Algorithm~\ref{alg:set_cover} and
   \begin{multline*}
      \cS_{k,i}' = \{ S_{\ell} \in \cS_{k,i} : |S_\ell \setminus C_{k+1}| \geq m^{1 - (i+1)\alpha} / 2 \\ \text{ and } |S_{\ell} \setminus C_{k+1}| / w_k \geq L/(1+\eps) \}.
   \end{multline*}
   Then $|\cS_{k,i}'| \leq |\cS_{k,i}| / 2m^{\mem/4}$ with probability at least $1 - m \exp\left( -m^{\mem/4} / 2 \right)$.
\end{lemma}
\begin{proof}
   Suppose we are currently processing group $\cX_{i,j}$.
   If at this point, the number of sets $S_{\ell} \in \cS_{k,i}$ with both $|S_{\ell} \setminus C| \geq m^{1 - (i+1)\alpha} / 2$ and $|S_{\ell} \setminus C| / w_k \geq L/(1+\eps)$
   is at most $|\cS_{i,k}| / 2m^{\mem/4}$ then we are done (since $C \subseteq C_{k+1}$).
   Otherwise, with probability at least $1 - \left( 1 - m^{\mem/2} / |\cS_{k,i}| \right)^{|\cS_{k,i}| / 2m^{\mem/4}} \geq 1 - \exp\left( - m^{\mem/4} / 2 \right)$,
   the group $\cX_{i,j}$ contains a set $S_{\ell}$
   with both $|S_{\ell} \setminus C| \geq m^{1 - (i+1) \alpha} / 2$ and $|S_{\ell} \setminus C| / w_k \geq L/(1+\eps)$.
   The algorithm then adds $S_{\ell}$ to the solution.
   Note that if this happens $2m^{(i+1) \alpha}$ times then we have covered the ground set and so the conclusion of the lemma holds trivially.
   Hence, we conclude that with probability at least $1 - m \exp\left( -m^{\mem/4} / 2 \right)$, we have $|\cS_{k,i}'| \leq |\cS_{k,i}| / 2m^{\mem/4}$.
\end{proof}

Let us now fix $\alpha = \mem/8$.
The following lemma shows that the potential function decreases geometrically after every iteration.
\begin{lemma}
   \label{lem:set_cover_potential}
   Fix an iteration $k$.
   Then $\Phi_{k+1} \leq \Phi_k / m^{\mem/8}$ with probability at least $1 - \frac{16m}{\mem} \exp\left( -m^{\mem/4 } / 2 \right)$.
\end{lemma}
\begin{proof}
   By Claim~\ref{claim:sc_failure}, an iteration fails with probability at most $\frac{m}{\alpha} \exp\left( -m^{\mem/2} \right) = \frac{8m}{\mem} \exp\left( -m^{\mem/2} \right)$.

   For $i \in [1/\alpha]$, let
   $\cA_i = \setst{ S_\ell \in \cS_{k,i} \cap \cS_{k,i}' }{ \frac{|S_{\ell} \setminus C_{k+1}|}{w_{\ell}} \geq \frac{L}{1+\eps}}$
   and $\cB_i = \setst{ S_\ell \in \cS_{k,i} \setminus \cS_{k,i}' }{ \frac{|S_{\ell} \setminus C_{k+1}|}{w_{\ell}} \geq \frac{L}{1+\eps}}$.
   Then
   \[
      \Phi_{k+1} = \sum_{i=1}^{1/\alpha} \Big(
      \sum_{S_\ell \in \cA_i} |S_\ell \setminus C_{k+1}| +
      \sum_{S_\ell \in \cB_i} |S_\ell \setminus C_{k+1}|
      \Big).
   \]
   By Lemma~\ref{lem:set_cover1} and the definition of $\cS_{k,i}'$, we have
   \[
      \sum_{S_\ell \in \cA_i} |S_\ell \setminus C_{k+1}|
      \leq |\cS_{k,i}'| m^{1-(i-1) \alpha} \leq |\cS_{k,i}| m^{1-(i-1) \alpha - \mem/4} / 2.
   \]
   On the other hand, we have
   $
      \sum_{S_\ell \in \cB_i} |S_\ell \setminus C_{k+1}|
      \leq |\cS_{k,i}| m^{1-(i+1)\alpha} / 2
   $.
   Plugging in $\alpha = \mem/8$, we have
   \[
      \Phi_{k+1} \leq \sum_{i=1}^{1/\alpha} |\cS_{k,i}| m^{1-i\mem/8} / m^{\mem/8}.
   \]
   Finally, we can also write
   \[
      \Phi_k = \sum_{i=1}^{1/\alpha} \sum_{S_\ell \in \cS_{k,i}} |S_\ell \setminus C_k|
      \geq \sum_{i=1}^{1/\alpha} |\cS_{k,i}| m^{1 - i\mem/8}.
   \]
   Hence, we conclude that $\Phi_{k+1} \leq \Phi_{k} / m^{\mem/8}$.
\end{proof}

We now show that the number of iterations of the inner while loop is quite small with high probability.
\begin{lemma}
   \label{lem:sc_rounds}
   Let $K = \inf \{k\!>\!0 : \Phi_k\!=\!0\}$.
   Then $K \leq \frac{18\log \Phi_0}{\mem \log m}$ with probability
   $\geq 1 - \frac{32m}{\mem} \exp\left( -m^{\mem/4} / 2 \right)$.
\end{lemma}
\begin{proof}
   Let us define a new random process $\Phi'_k$, which is coupled to $\Phi_k$ as follows. First, $\Phi_0' = \Phi_0$.
   For $k \geq 1$ we have two cases.
   If $\Phi_k \geq 1$ then we set
   \[
      \Phi_{k+1}' =
      \begin{cases}
         \Phi_{k}' / m^{\mem/8} & \text{if $\Phi_{k+1} \leq \Phi_{k} / m^{\mem/8}$} \\
         0                      & \text{otherwise}
      \end{cases}.
   \]
   On the other hand if $\Phi_k = 0$ then we set $\Phi_{k+1}' = \Phi_{k}' / m^{\mem/8}$ with probability $1 - \frac{16m}{\mem} \exp\left( -m^{\mem/4} / 2 \right)$.
   With the remainder probability, we set $\Phi_{k+1}' = \Phi_k$.

   Observe that with this coupling, we have $\Phi_{k} \leq \Phi_k'$ for all $k \geq 0$.
   Set $K' = \frac{18 \ln \Phi_0}{\mem \ln m}$.
   We will show that $\Phi'_{K'} < 1$ with probability at least $1 - \frac{32m}{\mem} \exp\left( -m^{\mem/4} / 2 \right)$.
   This then implies that $K \leq K'$ with the same probability since $\Phi_{K'}$ must be a nonnegative integer.

   Let us say iteration $k$ is \emph{good} if $\Phi_{k+1}' \leq \Phi_{k}' / m^{\mem/8}$.
   Otherwise, we say it is \emph{bad}.
   By Lemma~\ref{lem:set_cover_potential} and the definition of $\{\Phi_k'\}_k$, it follows that iteration $k$ is good with probability
   at least $1 - \frac{16m}{\mem} \exp\left( -m^{\mem/4} / 2 \right)$.
   After $\frac{8 \ln \Phi_0}{\mem \ln m} + 1 \leq \frac{9 \ln \Phi_0}{\mem \ln m}$ good iterations, we have $\Phi_k' < 1$ so it suffices to show
   that after $K'$ iterations there are at most $\frac{9 \ln \Phi_0}{\mem \ln m} = K'/2$ bad iterations.
   Indeed, after $K'$ iterations, the expected number of bad iterations is at most $K' \cdot \frac{16m}{\mem} \exp\left( -m^{\mem/4} / 2 \right)$ so by Markov's Inequality,
   there are more than $K'/2$ bad iterations with probability at most $\frac{32m}{\mem} \exp \left( -m^{\mem/4} / 2 \right)$.
   This completes the proof.
\end{proof}

We are now ready to prove the main theorem in this section.
Define $\Phi = \sum_{\ell \in [n]} |S_\ell|$.
\begin{theorem}
   Algorithm~\ref{alg:set_cover} returns $(1+\eps)H_{\Delta}$-approximate minimum set cover.
   Moreover, the inner loop is executed at most 
   \[
       O\left( \frac{\ln(\Phi) \log_{1+\eps}(\Delta w_{\max} / w_{\min})}{\mem \ln m} \right)
   \]
   times with probability
   at least 
   \[
       1 - \frac{64m}{\mem}\exp\left( -m^{\mem/4} / 2\right).
   \]
\end{theorem}
\begin{proof}
   The correctness follows because the algorithm implements the $\eps$-greedy algorithm for set cover.

   Let $M = \max_{\ell} \left\{ \frac{|S_\ell|}{w_{\ell}} \right\}$.
   We can prove the running time as follows.
   Let $K = \frac{18 \log \Phi}{\mem \log m}$ and consider splitting up the iterations of the inner while loop into blocks of size $K$.
   By block $t$, we will refer to iterations $tK, \ldots, (t+1)K - 1$ of the inner while loop.
   We say that a block $t$ is good if $L$ has decreased at least once during that block.
   If the algorithm has already completed by that point,
   we will instead just flip a coin which comes up heads with probability at least $1 - \frac{32m}{\mem} \exp\left( -m^{\mem/4} / 2 \right)$.
   If it comes up heads we will call that block good.
   Otherwise, we call the block bad.

   Note that after $\log_{1+\eps}(M w_{\max})$ good blocks, we are guaranteed that the algorithm has terminated.
   If we consider $2 \log_{1+\eps}(M w_{\max})$ blocks then it suffices that at most $\log_{1+\eps}(M w_{\max})$ blocks are bad.
   The expected number of bad blocks is $\frac{32m}{\alpha} \exp \left( -m^{\mem/4} / 2 \right)$ so applying Markov's Inequality shows that
   after $O\left( \frac{\ln(\Phi) \log_{1+\eps}(M w_{\max})}{\mem \ln m} \right)$ iterations, the algorithm has terminated
   with probability at least $1 - \frac{64m}{\alpha}\exp\left( -m^{\mem/4} / 2 \right)$.
   Using the trivial bound $M \leq \Delta / w_{\min}$ and replacing $\alpha$ with $\mem/8$
   completes the proof.
\end{proof}

\subsection{MapReduce Implementation}
\SectionName{logDeltaSCMR}

It is straightforward to implement most steps in Algorithm~\ref{alg:set_cover} in MapReduce.
However, we will highlight two nontrivial step here.
The first is how to propagate information such as the set of covered elements $C$ to all the machines.
To do this, it is convenient to imagine all the machines as arranged in an $O(m^{\mem})$-ary tree where the root of the tree is the central machine.
Then the central machine can pass $C$ down to its children.
These machines then pass down to their children and so on.
By doing this, all machines will know $C$ in $O\left( \frac{\ln n}{\mem \ln m} \right)$ MapReduce rounds.

The machines can also determine $|\cS_{k,i}|$ in a similar manner but starting at the leaves of tree.
Here, each machine will compute the number of sets that are in $\cS_{k,i}$ and send that quantity to its parents.
The parents then sum up the input and their own contribution to $|\cS_{k,i}|$ which they send to their own parents.
Eventually the root is able to compute $|\cS_{k,i}|$ and then propagates the number back down to the leaves as done above.
This again takes $O\left( \frac{\ln n}{\mem \ln m} \right)$ rounds in MapReduce.

We can also use a similar strategy to check that $|\cX_{i,j}|$ is small in Line~\ref{ln:alg_sc_fail} of Algorithm~\ref{alg:set_cover}.
We thus have the following theorem.

\begin{theorem}
\TheoremName{logDeltaSCMR}
   There is a MapReduce algorithm that returns a $(1+\eps) H_{\Delta}$-approximate minimum set cover
   which uses memory $O(m^{1 + \mem} \log n)$ and runs in $O\left( \frac{\ln(\Phi) \log_{1+\eps} (\Delta w_{\max} / w_{\min}) \ln(n)} {\mem^2 \ln^2(m)} \right)$ with probability
   at least $1 - O\left( \frac{m}{\mem} \right) \exp\left( -m^{\mem/4} / 2 \right)$.
\end{theorem}

Using the bound $\Phi \leq nm$ yields the bound given in \Figure{table}.

\begin{remark}
    Lemma~2.5 in \cite{BRS94} gives a simple way to preprocess the input so that $w_{\max} / w_{\min} \leq mn / \eps$ as follows.
    Let $\gamma = \max_{i \in [m]} \min_{S \ni i} w(S)$.
    This is a lower bound on the cost of any cover.
    First, we add every set with cost at most $\gamma \eps / n$ to the cover; this constributes a cost of at most $\gamma \eps \leq \eps \cdot OPT$.
    Next, we remove any set with cost more than $m \gamma$ since $OPT \leq m \gamma$.
    All these steps can be done using a broadcast tree in $O( \log(n) / (\mem \log(m)))$ rounds of MapReduce.
\end{remark}

\section{Maximum weight matching}
\SectionName{matching}

We have a graph $G = (V,E,w)$ but now $w \colon E \to \R$ is a weight function on the edges.
A matching in a graph is a subset $M \subseteq E$ such that $e_1 \cap e_2 = \emptyset$ for any distinct $e_1, e_2 \in M$.
The maximum weight matching in a graph is a matching $M$ that maximizes $\sum_{e \in M} w_e$.

\subsection{The local ratio method}

\begin{tcolorbox}[title=Sequential local ratio algorithm for maximum weight matching]
	Arbitrarily select an edge $e$ with positive weight and reduce its weight from itself and its neighboring edges.
    Push $e$ onto a stack and repeat this procedure until there are no positive weight edges remaining.
    At the end, unwind the stack adding edges greedily to the matching.
\end{tcolorbox}
If the edge $e = (u,v)$ was selected,
then the weight reduction makes the updates 
$w_{e^\prime} \gets w_{e^\prime} - w_e$ for any edge $e^\prime$ such that $e^\prime \cap e \neq \emptyset$. In contrast to the minimum vertex cover algorithm, weights in the graph can be reduced to negative weights.

\begin{theorem}[Paz-Schwartzman \cite{PS17}]
	\label{thm:lr_matching}
	The above algorithm returns a matching which is at least half the optimum value.
\end{theorem}

\comment{
\begin{proof}
	The output of the algorithm is clearly a matching.
    Again, we will prove the approximation guarantee via LP duality.
    
   The primal is
   \begin{equation}
      \label{eqn:p_matching}
      \max \sum_{e \in E} w_e x_e \quad \text{s.t.} \quad \sum_{e \ni v} x_e \leq 1 \,\, \forall v \in V, \quad x \geq 0.
   \end{equation}
   The dual is
   \begin{equation}
      \label{eqn:d_matching}
      \min \sum_{v \in V} y_v \quad \text{s.t.} \quad y_u + y_v \geq w_e \,\, \forall e = (u,v) \in E, \quad y \geq 0.
   \end{equation}
   
   Let $w'$ denote the edge weights of the graph as we unwind the stack.
   We will also initialize both the primal variables $x$ and the dual variables $y$ to 0.
   Note that for the weight vector $w'$, the variables are feasible for the corresponding LP.
   Whenever we pop an edge $e = (u,v)$ off the stack, we will increment $y_u, y_v, w_e'$ by $w^t_e$, where $w^t_e$ denotes
   the weight of $e$ when it was added to the stack.
   We will also add $e$ to the matching if possible.
   Note that if $y$ was feasible with the weights $w'$ before the update then it is also feasible after the update.
   Moreover, the objective value of the dual increases by exactly $2w^t_e$.
   
   Now let us look at the increase in the primal objective. If $e$ is added to the matching then the objective value increases by $w^t_e$.
   If $e$ is not added to the matching then there is an edge $e'$ incident to $e$ that is already in the matching.
   In which case, the weight of $e'$ increases by $w^t_e$ so the primal increases by at least $w^t_e$.
   
   Once the stack has been depleted, we have $w' = w$ so $y$ is feasible for the dual LP.
   Moreover, as we unwind the stack, the objective value of the primal increases by at least half the objective value of the dual,
   so this gives a 2-approximation to the maximum weight matching.
\end{proof}
}

\subsection{Randomized local ratio}

As in \Section{fSCRLR}, we apply our randomized local ratio technique to make the algorithm above amenable to parallelization.
For intuition, consider a fixed a vertex $v$ and suppose we sample approximately $n^{\mem}$ of its edges uniformly at random.
If $e$ is the heaviest sampled edge then there are only about $d(v) / n^{\mem}$ edges incident to $v$ that are heavier than $e$.
Hence, in the local ratio algorithm, if we choose $e$ as the edge to reduce then this effectively
decreases the degree of $v$ by a factor of $n^{-\mem}$.

\begin{algorithm}
    \caption{\small
    	$2$-approximation for maximum weight matching.
        Blue lines are centralized.
    }
    \label{alg:matching2apx}
    \begin{algorithmic}[1] 
        \Procedure{ApproxMaxMatching}{$G = (V,E)$}
        	\State $E_1 \gets E$, $d_{1}(v) \gets d(v)$, $i \gets 1$
            \State $S \gets \emptyset$ \Comment{Initialize an empty stack.}
            \While{$E_i \neq \emptyset$}
            	\For{each vertex $v \in V$}  \label{matching:for1}
                	\If{ $|E_i| < 4 \eta$ }
                    	\State Let $E'_v$ be all edges in $E_i$ incident to $v$
                    \Else
                		\State Construct $E'_v \subseteq E_i \intersect \delta(v)$ by sampling i.i.d.\ with probability
   			                   $p = \min\big\{\frac{\eta}{|E_i|},1 \big\}$ \label{ln:sample}
                               \label{matching:sample}
                    \EndIf
                \EndFor
                
                \If{$\sum_v |E_v'| > 8 \eta$}
                	\State \textbf{Fail}   \label{matching:fail}
                \EndIf
                
                \CFOR{each vertex $v \in V$}  \label{matching:for2}
                	\CSTATET Let $e \in E_v'$ be the heaviest edge and apply weight reduction to $e$  \label{ln:weight_reduction}
                    \CSTATE Push $e$ onto the stack $S$  \label{matching:push}
                \ENDFOR
                
                \State Let $E_{i+1}$ be the subset of $E_i$ with positive weights  \label{matching:Ei}
                \State $\rhd$: Let $d_{i+1}(v)$ denote $|\{e \in E_{i+1} : v \in e\}|$ \label{ln:matching_di}
                \State $i \gets i + 1$
            \EndWhile
            \State Unwind $S$, adding edges greedily to the matching $M$
            \State \textbf{return} $M$
        \EndProcedure
    \end{algorithmic}
\end{algorithm}

The following claim follows by a simple Chernoff bound.
\begin{claim}
	\label{claim:fail}
	For $i \geq 1$ and conditioned on iteration $i - 1$ not failing, iteration $i$ fails with probability at most $\exp(-\eta)$.
\end{claim}

\begin{lemma}
	\label{lem:first_phase}
    Suppose $\eta = n^{1+\mem}$ for some constant $\mem > 0$ and $|E| = n^{1+c}$ for some $c > \mem$.
	Then, with probability at least $1 - (n+1) \cdot \exp(-n^{\mem})$:
    \begin{itemize}
    	\item
        	the first iteration does not fail; and
        \item
        	$d_2(v) \leq n^{c}$ for all $v \in V$, where $d_2$ is as defined in Line~\ref{ln:matching_di} of Algorithm~\ref{alg:matching2apx}.
    \end{itemize}
\end{lemma}
\begin{proof}
	Let $k_v$ be the number of edges incident to $v$ with positive weight when we reach $v$ in the for loop in ~\Cref{ln:weight_reduction}.
    If $k_v \leq n^c$ then we are done so suppose $k_v > n^c$.
    The probability that we do not sample any of the heaviest $n^c$ edges that are currently
    incident to $v$ is at most $\left( 1 - \frac{n^{1+\mem}}{|E|} \right)^{n^c} \leq \exp\left(
    -\frac{n^{1+\mem+c}}{|E|} \right) \leq \exp(-n^{\mem})$.
    Combining with Claim~\ref{claim:fail} and taking a union bound completes the proof.
\end{proof}

In the subsequent analysis, we will assume that $\mem$ is a positive constant.
(Actually, it suffices to take $\mem = \omega(\log\log n/\log n)$.)

\begin{lemma}
	\label{lem:phases}
    Suppose $\eta = n^{1+\mem}$ for any constant $\mem>0$.
	Let $\Delta_i = \max_v d_i(v)$.
    For $i > 2$ and conditioning on the past $i-1$ iterations not failing, with probability at least
    $1 - (n+1) \cdot \exp(-n^{\mem/2})$:
    \begin{itemize}
    	\item
    	iteration $i$ does not fail; and
        \item
        $\Delta_{i+1} \leq \Delta_{i} / n^{\mem/4}$.
    \end{itemize}
\end{lemma}
\begin{proof}
	Let $k_v$ be the number of edges incident to $v$ with positive weight when we reach $v$ in the for loop in ~\Cref{ln:weight_reduction}.
    If $k_v \leq \Delta_i / n^{\mem/4}$ then we are done so suppose $k_v > \Delta_i / n^{\mem/4}$.
    The probability that we do not sample any of the heaviest $k_v/n^{\mem/4}$ edges that are
    currently incident to $v$ is at most $\left( 1 - \frac{n^{1+\mem}}{|E|} \right)^{k_v/n^{\mem/4}}
    \leq \exp\left( -\frac{n^{1+\mem/2} \Delta_i}{|E|} \right) \leq \exp(-n^{\mem/2})$.
    Hence, we have $d_{i+1}(v) \leq k_v/n^{\mem/4} \leq d_i(v) / n^{\mem/4}$
    with probability at least $1 - \exp(-n^{\mem/2}) = 1 - \exp( - \omega(\log n) )$.
    Combining with Claim~\ref{claim:fail} and taking a union bound completes the proof.
\end{proof}

\begin{theorem}
	\TheoremName{matchingLargeMem}
    Suppose $\eta = O(n^{1+\mem})$ for any constant $\mem>0$.
	With probability $1 - O(cn/\mem) \exp(-n^{\mem})$, \Cref{alg:matching} terminates in
    $O(c/\mem)$ iterations and returns a 2-approximate maximum matching.
\end{theorem}
\begin{proof}
	By Lemma~\ref{lem:first_phase}, the first iteration does not fail and $d_2(v) \leq n^c$ for all $v \in V$.
    By Lemma~\ref{lem:phases}, after at most $O(c/\mem)$ iterations, the algorithm has not failed and
    the total number of edges with positive weight remaining is at most $8n^{1+\mem}$.
    At this point, the algorithm completes its last iteration of weight reduction then unwinds the stack and returns a matching.
    The correctness of the algorithm follows from \Theorem{lr_matching}.
\end{proof}

The preceding analysis assumes that $\mem = \omega(\log \log n / \log n)$.
In \Appendix{matchingSmallMem} we additionally handle the case $\mem=0$ (or equivalently, $\mem=\Theta(1/\log n)$).

\subsection{MapReduce implementation}

\begin{theorem}
\TheoremName{matchingMR}
There is a MapReduce algorithm that computes a 2-approximation to the maximum weight matching using $O(n^{1+\mem})$ space per machine and
$O(c / \mem)$ rounds (when $\mem = \omega(\log \log n / \log n)$) or $O(\log n)$ rounds (when $\mem = 0$).
\end{theorem}

\begin{proof}
	As a sequential algorithm, the correctness of \Algorithm{matching} is established by \Theorem{matchingLargeMem} in the case $\mem$ is constant and \Theorem{matchingSmallMem} when $\mem=0$.
    We now show how to parallelize \Algorithm{matching}.
    
    As in \Theorem{fSCMR}, there are $n^{c-\mem}$ machines and each edge is assigned to one of the machines with $O(n^{1+\mem})$ edges per machine.
    Each vertex (and its adjacency list) is randomly assigned to one of the machines, so the space per machine is $O(n^{1+\mem})$ whp.
    Each edge $e$ stores its original weight and maintains a bit indicating if $e \in E_i$.
    
    The local ratio steps (lines \ref{matching:for2}-\ref{matching:push}) are performed sequentially on the central machine.
    The input to the central machine is the sets $E_v'$, together with the original weights of those edges.
    The total size of the input is proportional to $\sum_v \card{E_v'}$, which is guaranteed to be $O(\eta)$ by line \ref{matching:fail}.
    The central machine is stateful: it maintains values $\phi(v)$ for each vertex $v$, initially zero.
    The value $\phi(v)$ will always equal the total value of the weight reductions for all edges incident to $v$.
    Thus, for any edge $e=\set{u,v}$ that was never added to the stack, if its original weight is $w_e$, then its modified weight is $w_e-\phi(u)-\phi(v)$.
    The weight reduction operation for edge $e=\set{u,v}$ is then straightforward: simply decrease both $\phi(u)$ and $\phi(v)$ by $w_e-\phi(u)-\phi(v)$.
    
    The main step requiring parallelization is line \ref{matching:Ei}.
    After the for-loop terminates, the central machine sends $\phi(v)$ to each vertex $v$, and informs each edge whether it was added to the stack. If so, the edge's modified weight is definitely non-positive so it cannot belong to $E_{i+1}$.
    Afterwards, each vertex $v$ then sends $\phi(v)$ to each edge $e \in \delta(v)$.
    Each edge $e=\set{u,v}$ receives $\phi(u)$ and $\phi(v)$ and, if it was not added to the stack, computes its modified weight; if this is positive then $e$ belongs to $E_{i+1}$.
    Given the knowledge of whether $e \in E_i$ and given $\card{E_i}$ (which is easy to compute in parallel), an edge can independently perform the sampling operation on line \ref{matching:sample}.
\end{proof}
\paragraph{$b$-matching.}
Using similar techniques, one can obtain a $(3-2/b+\eps)$-approximation to $b$-matching; see \Appendix{bmatching}.

\comment{
\begin{proof}
	The while loop of algorithm consists of three phases: sampling (lines 5-9), local ratio (lines 10-12), and pruning (lines 13-14).
    
    Assuming the number of edges is available from the previous iteration, the sampling phase can be
    done in a single MapReduce round similar to \Cref{alg:minvc}. With the reducer sampling the
    edges adjacent to a node $u$ with probability $n^{1+\mem} / |E_i|$.
    
    The local ratio phase can be done in a single reduce round as the number of sampled edges is
    $O(n^{1+\mem})$ w.h.p. During the local ratio phase, edges adjacent to any edge added to the stack theirs their weight decreased. To propagate the weight reduction, we sum up for each vertex $v$ the total decrease $\phi(v)$ caused by edges added to the stack adjacent to $v$. This total weight reduction can then be sent to all its neighbours. 
    
    Using the weight reduction, we can prune the edges of the graph in a single MapReduce round by collecting all the edges adjacent to node $v$ and decreasing their edge weights by $\phi(v)$. We then remove any edges that have non-positive weight. This can be implemented in the same way as the pruning phase in \Cref{alg:minvc}.
    
    Each iteration of the while loop takes at most 4 MapReduce rounds. A final round is needed to unwind the stack and produce the matching $M$.
\end{proof}
}

\section{Vertex and edge colouring}
\SectionName{colouring}

In this section, we show how to colour a graph with maximum degree $\Delta$ using $(1+o(1))\Delta$ colours in a constant number of rounds. 
As is the case with MIS, $(\Delta+1)$-vertex colouring is one of the most fundamental problems of distributed computing. 
In the CREW PRAM model, both MIS and $(\Delta+1)$-vertex colouring have algorithms that can be easily translated to $O(\log n)$-round algorithms in the MapReduce model. 
Luby's randomized algorithms for both MIS~\cite{L85} and $(\Delta+1)$-vertex colouring~\cite{L88} have clean MapReduce implementations by using one machine per processor in the CREW PRAM algorithm. Within the CREW PRAM and LOCAL model, both MIS and $(\Delta+1)$-vertex colouring have well established lower bounds. However, we are unaware of non-trivial (i.e. non-constant) lower bounds on round complexity in the MapReduce model. 

Although our algorithm does not use the randomized local ratio or the hungry-greedy paradigm developed in the previous sections, we feel that it is of independent interest as it is the first constant round algorithm for $\Delta+o(\Delta)$-vertex colouring within the MapReduce model that we are aware of. 

Recall that the number of edges is $n^{1+c}$, so $\Delta \geq n^c$ since the maximum degree is at least the average degree.
Our algorithm is very simple; first we randomly partition the vertex set into $\kappa \coloneqq n^{(c - \mem)/2}$ groups.
Within each group, the maximum degree is $(1+o(1))\Delta / \kappa$ with high probability, so $(1+o(1))\Delta / \kappa + 1$ colours suffices to colour each group.
The colour of a vertex is determined by the group that it is in and its colour within each group.
Hence, this gives a colouring with $(1+o(1))\Delta$ colours. Furthermore, we show that the subgraph induced by each group has a small number of edges. As a corollary we get a MapReduce algorithm for $(1 + o(1))\Delta$-colouring.

\begin{algorithm}
    \caption{$(1+o(1))\Delta$-vertex colouring}
    \label{alg:vertex_colouring}
    \begin{algorithmic}[1] 
        \Procedure{VertexColouring}{Graph $G = (V,E)$}
        \State Randomly partition $V$ into $\kappa$ groups, $V_1, \ldots, V_{\kappa}$
        \State Define $E_i = E[V_i]$ and $G_i = (V_i, E_i)$
        \If{$\exists i$ such that $|E_i| > 13n^{1+\mem}$}
            \State \textbf{Fail}
        \EndIf
        \For{every vertex $v$ in parallel}
            \State if $v \in V_i$ then send $N(v) \cap V_i$ to central machine $i$
        \EndFor
        \For{every central machine $i$ in parallel}
            \State Let $\Delta_i$ be max degree of $G_i$
            \State Colour $G_i$ using the standard $(\Delta_i + 1)$-vertex colouring algorithm
            \For{every $v \in V_i$}
               \State Let $c_i(v)$ be colour of $v \in V_i$ in this colouring
               \State \textbf{Output} $(i, c_i(v))$ as colour for $v$.
            \EndFor
        \EndFor
        \EndProcedure
    \end{algorithmic}
\end{algorithm}


\begin{lemma}
  \label{lem:col_degree}
  For all $\mem > 0$ and $\kappa = n^{(c - \mem)/2}$ then we have $\Delta_i \leq (1+ n^{-\mem/2} \sqrt{6 \ln n})\Delta / \kappa$ for all $i$ with probability at least $1 - 1/n$.
\end{lemma}
\begin{proof}
   For a single vertex $v$ of degree $d$, the probability that $v$ has degree greater than $(1+\eps) \Delta / \kappa$ (in the subgraph induced by its group) is at most $\exp\left(-\eps^2 \frac{\Delta}{3 \kappa}\right)$ by a standard Chernoff bound (Theorem~\ref{thm:chernoff}).
   Since $\Delta / \kappa \geq n^{c/2 + \mem/2} \geq n^\mem$, we may take $\eps = n^{-\mem/2} \sqrt{6 \ln n}$ in which case the probability that a vertex $v$ has degree greater than $(1+\eps)\Delta/\kappa$ is at most $1/n^2$.
   Taking a union bound then yields the lemma.
\end{proof}

\begin{lemma}
   We have $|E_i| \leq 13n^{1+\mem}$ with probability at least $1 - n^2 \cdot  \exp\left( -n^{\mem} \right)$.
\end{lemma}
\begin{proof}
   By the Hajnal-Szemer\'{e}di Theorem (Theorem~\ref{thm:hajnal}) applied to the line graph of $G$, we may partition the edges into $2\Delta$ sets $F_1, \ldots, F_{2\Delta}$ such that the following holds for all $j \in [2\Delta]$:
   \begin{itemize}
      \item $|F_j| \geq n^{1+c} / (4\Delta)$; and
      \item if $e, e' \in F_j$ and $e \neq e'$ then $e \cap e' = \emptyset$.
   \end{itemize}
   
   For analysis, fix $E_i$ and an edge class $F_j$. For each $e \in F_j$, let $X_{j,e}$ be the indicator random variable which is 1 if edge $e$ ends up in $E_i$.
   Then $\E \sum_{e \in F_j} X_{j,e} = |F_j| / \kappa^2$.
   Since any distinct $e, e' \in F_j$ do not share a common vertex, it follows that $\{X_{j,e}\}_{e \in F_j}$ are mutually independent random variables.
   Therefore, we may apply a Chernoff bound (Theorem~\ref{thm:chernoff}) to get
   \begin{multline*}
      \Pr\Big[\sum_{e \in F_j} X_{j,e} > 13 |F_j| / \kappa^2\Big]
      \leq
      \exp\left( -12 \frac{|F_j|}{3\kappa^2} \right) \\
      \leq
      \exp\left( -\frac{n^{1+c}}{\Delta \kappa^2} \right)
      \leq 
      \exp\left( -\frac{n^{c}}{\kappa^2} \right)
      =
      \exp\left( -n^{\mem} \right).
   \end{multline*}
   Taking a union bound over all $i$ and $j$ gives the claim.
\end{proof}

\begin{corollary}
   Algorithm~\ref{alg:vertex_colouring} returns a $\left( 1 + n^{-\mem/2} \sqrt{6 \ln n} + n^{-\mem} \right) \Delta$-colouring of $G$ with high probability.
\end{corollary}
\begin{proof}
   We can colour each $G_i$ with $\Delta_i + 1$ colours using the standard greedy colouring algorithm. 
   By Lemma~\ref{lem:col_degree} $\Delta_i \leq (1+ n^{-\mem/2} \sqrt{6 \ln n})\Delta / \kappa$ with probability at least $1 - 1/n$.
   Hence, we use at most $\kappa(\Delta_i + 1) \leq \left( 1 + n^{-\mem/2} \sqrt{6 \ln n} + n^{-\mem} \right) \Delta$ colours.
\end{proof}

\begin{theorem}
\label{thm:vtxcolouring}
If $\mem = \omega(\log \log n / \log n)$ and the memory on each machine is $O(n^{1 + \mem})$ then there is a MapReduce algorithm for $(1 + o(1))\Delta$-colouring which succeeds with high probability.
\end{theorem}

\begin{remark}
$(1+o(1))\Delta$-edge colouring can be achieved with almost the same algorithm, partitioning the edges into groups instead of vertices and colouring the groups with the algorithm of Misra and Gries~\cite{MG92}.
\end{remark}

\begin{theorem}
\label{thm:edgecolouring}
There are constant-round MapReduce algorithms for $(1+o(1))\Delta$-vertex colouring and $(1+o(1))\Delta$-edge colouring that succeed w.h.p.
\end{theorem}

\appendix
\clearpage

\section{Improved algorithm for maximal independent set}
\AppendixName{improvedMIS}

In this section, we discuss how to improve the algorithm in \Section{MIS} to obtain an algorithm that rounds in $O(c/\mem)$ rounds instead of $O(1/\mem^2)$.
The basic idea is that if we sample $n^{(i+1)\alpha}$ vertices from the set of vertices that have degree $[n^{1-i\alpha}, n^{1-(i-1)\alpha})$ and proceed as in Algorithm~\ref{alg:mis}, then the degree of almost all vertices decreases by a factor of $n^{\alpha}$.
In particular, this implies that the number of edges decrease by a factor of $n^{\alpha}$ which gives a algorithm using $O(c/\alpha)$ rounds.
We will set $\alpha = \mem/8$ in the analysis below.

\begin{algorithm}
    \caption{Improved algorithm for maximal independent set. Blue lines are centralized.}
    \label{alg:mis_v2}
    \begin{algorithmic}[1] 
        \Procedure{MIS2}{$G = (V,E)$}
        	\State $I_1 \gets \{ v \colon d(v)=0 \}, E_1 \gets E, V_1 \gets V, k \gets 1$
            \While{$|E_k| \geq n^{1+\mem}$}
            	\State $I \gets I_k$
                \State Let $V_{k,i} = \{ v \in V_k \colon n^{1-i\alpha} \leq d_I(v) < n^{1-(i-1)\alpha} \}$, for $i \in \{1,\ldots,1/\alpha\}$.
            	\For{$i = 1, \ldots, 1/\alpha$}  \label{mis2:for1}
                	\State Draw $n^{(i+1)\alpha}$ groups of $n^{\mem/2}$ vertices from $V_{k,i}$, say $\cX_{i,1}, \ldots, \cX_{i,n^{(i+1)\alpha}}$  \label{mis2:draw}
                \EndFor
                \CFOR{$i = 1, \ldots, 1/\alpha$}  \label{mis2:for2}
                	\CFOR{$j = 1, \ldots, n^{(i+1)\alpha}$}
                      \CIFT{$\exists v_{i,j} \in \cX_{i,j}$ such that $d_I(v_{i,j}) \geq n^{1-(i+1)\alpha}$}
                          \CSTATE $I \gets I \cup \{v_{i,j}\}$  \label{mis2:updateI}
                      \ENDIF
                    \ENDFOR
                \ENDFOR
                \State $I_{k+1} \gets I, V_{k+1} \gets \{v \colon d_{I_{k+1}}(v) > 0\}, E_{k+1} \gets E[V_{k+1}]$
                \State $k \gets k + 1$
            \EndWhile
            \State Find maximal independent set in $(V_k, E_k)$ and add it to $I$ \Comment{Total edges is $< n^{1+\mem}$}
            \State \textbf{return} $I$
        \EndProcedure
    \end{algorithmic}
\end{algorithm}

In the algorithm, the set $V_{k,i}$ corresponds to the set of vertices in iteration $k$ with degree between $n^{1-i\alpha}$ and $n^{1-(i-1)\alpha}$.
We first show that most vertices in $V_{k,i}$ lose a significant number of their neighbours.
This will be important to show good progress in removing the edges.
\begin{lemma}
	\label{lem:Vki}
	Let $V_{k,i}$ be as in the algorithm and $V_{k,i}' = \{v \in V_{k,i} : d_{I_{k+1}}(v) \geq n^{1 - (i+1)\alpha}\}$.
   In other words, $V_{k,i}'$ is the set of vertices within $V_{k,i}$ that did not decrease their degree by at least a factor of $n^{i\alpha}$.
   Then $|V_{k,i}'| \leq |V_{k,i}| / n^{\mem/4}$ w.h.p.
\end{lemma}
\begin{proof}
	Suppose we are currently processing group $\cX_{i,j}$.
    If the number of vertices $v$ in $V_{k,i}$ with $d_I(v) \geq n^{1-(i+1)\alpha}$ is at most
    $|V_{k,i}|/n^{\mem/4}$ then we are done.
    Otherwise, with probability at least $1-(1-n^{-\mem/4})^{n^{\mem/2}} \geq 1 - \exp(-n^{\mem/4})$, the group $\cX_{i,j}$ contains a vertex $v_{i,j}$ with $d_I(v_{i,j}) \geq n^{1-(i+1)\alpha}$ and we add it to the independent set.
    This can only happen at most $n^{(i+1)\alpha}$ times because if it happens $n^{(i+1)\alpha}$ times then there are no more vertices in the graph. Since we sampled precisely $n^{(i+1)\alpha}$ groups of vertices, we expect $|V_{k,i}'| \leq |V_{k,i}| / n^{\mem/4}$ w.h.p.
\end{proof}
By choosing $\alpha = \mem/8$, we can guarantee that every vertex class $V_{k,i}$ decrease their degree by a factor of $n^{\alpha}$ with high probability. This gives the following lemma.
\begin{lemma}
	We have $|E_{k+1}| \leq 2|E_k|/n^{\mem/8}$ w.h.p.
\end{lemma}
This lemma implies that after $O(c/\mem)$ rounds, we can fit all the data onto a single machine.
\begin{proof}
	Let $S \subseteq V$ and define $d_I(S) = \sum_{v\in S} d_I(v)$. 
	We will lower bound $d_{I_k}(V_{k,i})$ and upper bound $d_{I_{k+1}}(V_{k,i})$.
    We have $d_{I_k}(V_{k,i}) \geq |V_{k,i}| n^{1-i\alpha}$.
    On the other hand,
    \begin{align*}
    d_{I_{k+1}}(V_{k,i})
     &~\leq~ |V_{k,i}'| n^{1-(i-1)\alpha} + |V_{k,i} \setminus V_{k,i}'| n^{1-(i+1)\alpha} \\
     &~\leq~ |V_{k,i}| n^{-\mem/4} n^{1-(i-1)\alpha} + |V_{k,i}| n^{1-(i+1)\alpha} \\
     &~=~ 2 |V_{k,i}| n^{1-(i+1)\mem/8},
    \end{align*}
    by Lemma~\ref{lem:Vki} and using $\alpha = \mem/8$.
    Since the sets $V_{k,i} \cap V_{k+1}$ form a partition of $V_{k+1}$, we have
    \begin{align*}
    2 |E_{k+1}|
      &= \sum_i d_{I_{k+1}}(V_{k,i} \cap V_{k+1}) \\
      &= \sum_i d_{I_{k+1}}(V_{k,i})
      \leq 2 \sum_i |V_{k,i}| n^{1-(i+1)\mem/8} \\
      &\leq 2 \sum_i d_{I_k}(V_{k,i}) n^{-\mem/8} \\
      & = 4 |E_k| n^{-\mem/8},
    \end{align*}
    proving the claim.
\end{proof}

\begin{theorem}
\TheoremName{MIS2}
There is a MapReduce algorithm to find a maximal independent set in $O(c/\mem)$ rounds
and $O(n^{1+\mem})$ space per machine, w.h.p.
\end{theorem}
\begin{proof}[Proof sketch]
The pseudocode in \Algorithm{mis_v2} is written in the manner of a sequential algorithm, but it is easy to parallelize.
The main step that is performed in parallel is the sampling (lines \ref{mis2:for1}-\ref{mis2:draw}).
The sets $\cX_{i,j}$ are then sent to a central machine.
Lines \ref{mis2:for2}-\ref{mis2:updateI} are performed sequentially on the central machine.
All other steps are straightforward to perform in the MapReduce model.

The space usage on the central machine is dominated by the total size of the neighbourhoods of all vertices in the sets $\cX_{i,j}$.
Thus, the required space per machine is proportional to
\begin{align*}
\sum_{i=1}^{1/\alpha} \: \sum_{j=1}^{n^{(i+1)\alpha}} \sum_{v \in \cX_{i,j}} d_I(v)
 & \leq
 \sum_{i=1}^{1/\alpha} n^{(i+1)\alpha} \cdot n^{\mem/2} \cdot n^{1-(i-1)\alpha} \\
 & = (1/\alpha) \cdot n^{1+\mem/2+2\alpha}.
\end{align*}
Since $\alpha=\mem/8$, this gives the claimed space bound.
\end{proof}

\section{Maximal clique}
\AppendixName{maxclique}

We can compute a maximal clique via a reduction to MIS.
This is, however, nontrivial because naively complementing the graph could result in a graph that is too large to fit into memory.
To resolve this, we will use a relabeling scheme which satisfies the following properties:
\begin{itemize}
   \item if there are $k$ active vertices then every active vertex has a label in $[k]$.
   \item every vertex knows $k$ and the label of all its neighbours.
\end{itemize}
Suppose we had such a labelling and consider a vertex $v$.
Suppose $N$ is the set of active neighbours of $v$.
Then its degree in the complement graph is $k - |N|$ and its active neighbourhood in the complement graph is $[k] \setminus N$.
Observe that, while the complement graph could have $\Omega(n^2)$ edges, there is no space issue because
each round only requires us to compute $O(n^{1+\mem})$ edges of the complement graph \emph{in total}.

We now describe the relabelling scheme.
Initially, the two conditions are trivially satisfied if we assume each vertex has a unique label in $[n]$.
To maintain the invariant, we add the following relabeling procedure:
\begin{enumerate}
   \item
      Suppose there are $k$ active vertices (the central machine always knows which vertices are active).
      Pick a permutation $\sigma \colon V \to [n]$ such that if $v$ is active then $\sigma(v) \in [k]$ and if $v$ is inactive then $\sigma(v) > k$.
      Send $\sigma(v)$ and $k$ to vertex $v$.
   \item
      Each vertex $v$ knows $\sigma(v)$.
      If $u$ is a neighbour of $v$ then $v$ queries $u$ for $\sigma(u)$.
      It then knows which of its neighbours are still active (since it knows $k$).
\end{enumerate}

\begin{corollary}
	\CorollaryName{maxclique}
   There is a MapReduce algorithm to find a maximal clique in $O(1/\mem)$ rounds w.h.p.
\end{corollary}

\section{Matching with $O(n)$ space per machine}
\AppendixName{matchingSmallMem}

In this section we discuss the matching algorithm (\Algorithm{matching}) in the case $\mem=0$, which corresponds to having $O(n)$ space per machine.

\begin{lemma}
	Suppose $\eta = n$.
    Then none of the iterations fail w.h.p.
    Moreover, conditioned on none of the iterations failing, we have
	$\E\left[|E_{i+1}| \mid E_i \right] \leq 0.975 |E_i|$.
\end{lemma}
\begin{proof}
	The first part of the lemma is by Claim~\ref{claim:fail}.
    
    Let $H_i = \{v : d_i(v) \geq |E|/n\}$ and $L_i = V \setminus H_i$.
    We call $H_i$ the heavy vertices at iteration $i$ and $L_i$ the light vertices at iteration $i$.
    Let $c = 0.9$.
    We claim that if $v$ is heavy then with probability at least $1/2$, we have
    $d_{i+1}(v) \leq c\cdot d_i(v)$.
    Let $k_v$ be the number of edges incident to $v$ with positive weight when we reach $v$ in the for loop in \Cref{ln:weight_reduction}.
    If $k_v \leq c \cdot d_i(v)$ then we are done so suppose $k_v > c \cdot d_i(v)$.
    The probability that we do not sample any of the heaviest $ck_v$ edges that are currently incident to $v$ is at most $\left( 1 - \frac{n}{|E|} \right)^{ck_v} \leq \exp\left( - c^2 \right) < 1/2$.
    
    Observe that $\sum_{v \in L_i} d_i(v) \leq |E_i|$ so $\sum_{v \in H_i} d_i(v) \geq |E_i|$.
    Hence, conditioned on $E_{i}$, we have
    \begin{eqnarray*}
    	\E |E_{i+1}| &\leq& \frac{1}{2} \left (\sum_{v \in L_i} d_i(v) +
        \sum_{v \in H_i} \left[ \frac{1}{2} d_i(v) + \frac{1}{2} cd_i(v) \right] \right) \\
        &\leq&
        \frac{3+c}{4} |E_i| = 0.975 |E_i|,
    \end{eqnarray*}
    proving the second part of the lemma.
\end{proof}

\begin{theorem}
	\TheoremName{matchingSmallMem}
	With probability at least $1 - 1/n$, Algorithm~\ref{alg:matching2apx} terminates in $O(\log n)$ iterations and returns a 2-approximate maximum matching.
\end{theorem}
\begin{proof}
	Let us consider a modified version of Algorithm~\ref{alg:matching2apx} where the condition to terminate when $|E_i| < 8\eta$ is removed.
    Let $E_i'$ denote the sequence of edges in the modified algorithm and $E_i$ to denote the sequence of edges in the unmodified version of Algorithm~\ref{alg:matching2apx}.
    There is a simple coupling which ensures that:
    \begin{itemize}
    	\item $E_1' = E_1$;
        \item $E_{i+1}' = E_{i+1}$ if $|E_{i}| \geq 8n$; and
        \item if $|E_i| < 8n$ then Algorithm~\ref{alg:matching2apx} terminates in iteration $i+1$.
    \end{itemize}
    Observe that the final condition implies that
    \[
       \Pr[\text{Algorithm~\ref{alg:matching} terminates by iteration $\tau+1$}] \geq \Pr[|E_{\tau}'| < 8n].
    \]
    Choose $\tau = 200\log(n) > \log(n^2) / \log(1/0.975)$.
    By the previous lemma, we have $\E |E_{\tau}'| \leq 1$ so by Markov's Inequality we have
    $\Pr[|E_{\tau}'| \geq 8n] \leq 1/(8n)$.
    In particular, Algorithm~\ref{alg:matching} terminates in $O(\log(n))$ iterations with probability at least $1 - 1/n$.
    
    The correctness follows from \Theorem{lr_matching}.
\end{proof}

\comment{

\section{$O(\log m)$-approximation for Unweighted Set Cover}
\SectionName{logmSCNick}

Following the notation of the previous section,
we wish to solve $\min \setst{ \card{X} }{ S(X) = U = [m] }$,
the unweighted set cover problem.
Let $OPT$ denote a fixed optimal solution to the set cover problem.
It is well-known that the natural greedy algorithm produces a set
$X$ with $S(X)=U$ and $\card{X} \leq (1+\ln m) \card{OPT}$.

To obtain our result in the MapReduce model, we will first explain
an alternative way to derive a sequential algorithm obtaining a comparable result.

\begin{theorem}
Assume $\eps \leq 1/2$.
Let $X$ be the output of \textsc{ApproxSC2}.
Then $S(X) = U$ and $\card{X} \leq \ceil{(1+4\eps) \ln(m) \card{OPT}}$.
\end{theorem}

\begin{algorithm}
    \caption{An $O(\log m)$-approximation for minimum set cover}
    \label{alg:minsc2}
    \begin{algorithmic}[1] 
        \Procedure{ApproxSC2}{$S_1,\ldots,S_n \subseteq U = [m]$}
            \Procedure{MaxCoverage}{$k, ~\eps$}
                \State $X_0 \leftarrow \emptyset$
                \For{$i=0,\ldots,k-1$}
                    \State Pick any $e_{i+1} \not\in X_{i-1}$ with $\card{S_{e_{i+1}} \setminus S(X_i)}
                        \geq (1-\eps) \max_x \, \card{S_x \setminus S(X_i)}$
                        \label{alg:unwtdSCmax}
                    \State Let $X_{i+1} \leftarrow X_i \union \set{e_{i+1}}$
                \EndFor
            \EndProcedure
            \State Let $\cX$ be an empty list
            \State Let $\alpha = \ln(m)/(1-\eps)$
            \ForAll{values $L$ of the form $(1+\eps)^i$ in the interval $[1, (1+\eps)m]$}
                \label{alg:unwtdSCforall}
                \State Let $X \leftarrow \textsc{MaxCoverage}(\ceil{\alpha L}, \eps)$
                \State If $S(X)=U$, add $X$ to the list $\cX$
            \EndFor
            \State \textbf{return} an $X \in \cX$ that minimizes $\card{X}$
        \EndProcedure
    \end{algorithmic}
\end{algorithm}

\textsc{MaxCoverage} implements the so-called ``$\epsilon$-greedy algorithm''
for the Maximum Coverage problem $ \max \setst{ \card{S(X)} }{ \card{X} \leq k } $,
meaning that, in line~\ref{alg:unwtdSCmax}, an $\epsilon$-approximate maximizer is chosen.

\begin{proof}
Obviously one of the values $L$ chosen in \textsc{ApproxSC2}
will satisfy $\card{OPT} \leq L \leq (1+\eps) \card{OPT}$.
Let us call this value $L^*$.
Consider the invocation of \textsc{MaxCoverage} with $k=\ceil{\alpha L^*}$,
and let $X_g$ denote its output.  Obviously $\card{X_g} \leq \ceil{\alpha L^*}$.

It remains to show that $S(X_g)=U$.
Fix $X^* \in \argmax \setst{ \card{S(X)} }{ \card{X} \leq L^* } $,
and note that the cardinality is bounded by $L^*$ not by $k$.
A small adjustment to the standard analysis of the Maximum Coverage problem yields
$$
\card{S(X^*)} - \card{S(X_i)}
    ~\leq~ \Big(1-\frac{1-\eps}{L^*}\Big)^i \cdot \card{S(X^*)} \qquad\forall i \geq 0.
$$
By definition of $L^*$, we have $S(X^*) = U = [m]$.
In particular, setting $i=k =\ceil{\alpha L^*}$, the output $X_g$ of \textsc{MaxCoverage} satisfies
$$
m - \card{S(X_g)}
    ~\leq~ \Big(1-\frac{1-\eps}{L^*}\Big)^{\ceil{\alpha L^*}} \cdot m
    ~<~ \exp(-\ln m) \cdot m ~=~ 1. 
$$
By integrality, we must have $S(X_g) = U$.
Therefore \textsc{ApproxSC2} adds this set $X_g$ is to the list $\cX$.

It follows that final output of \textsc{ApproxSC2} 
is a set $X$ with $S(X)=U$ and $\card{X} \leq \ceil{\alpha L^*}$.
\end{proof}

\subsection{MapReduce implementation}

Now we explain how how to design an MapReduce algorithm for the set cover problem.
Let $\Delta = \max_i \card{S_i}$.
The \textsc{GreedyScalingMR} algorithm of Kumar et al.~\cite{KMVV15}
can implement the \textsc{MaxCoverage} procedure in
$O(\log(\Delta)/\epsilon \delta)$ rounds
and $O(k \Delta n^\delta \log n)$ space per machine, for any $\delta>0$.
To implement \textsc{ApproxSC2}, we execute
all iterations of loop on line~\ref{alg:unwtdSCforall} in parallel.

\begin{theorem}
If the amount of space per machine is $O(\card{OPT} \Delta n^\delta \log n)$ 
then the unweighted set cover problem can be solved
in $O(\log(\Delta)/\epsilon \delta)$ rounds.
\end{theorem}

\section{Under contstruction: $O(\log m)$-approximation for Weighted Set Cover}

As in the previous section, we wish to solve the problem $\min \setst{ w(X) }{ S(X) = U }$.
It is well-known that the natural greedy algorithm produces a set
$X$ with $S(X)=U$ and $w(X) \leq (1+\ln m) w(OPT)$.

To obtain our result in the MapReduce model, we will first explain
an alternative sequential algorithm that obtains a comparable result.
The algorithm builds on the techniques of Khuller, Moss and Naor \cite{KMN99}
for solving the ``budgeted coverage problem'' $\max \setst{ \card{S(X)} }{ w(X) \leq L }$.

\begin{theorem}
Let $X$ be the output of \textsc{ApproxSC2}. Then
$S(X) = U$ and $w(X) \leq (1+\eps) (1+\ln m) L$.
\end{theorem}

\begin{algorithm}
    \caption{An $O(\log m)$-approximation for minimum weight set cover}
    \label{alg:minsc2}
    \begin{algorithmic}[1] 
        \Procedure{ApproxSC2}{$S_1,\ldots,S_n \subseteq U = [m], ~w_1,\ldots,w_n \in \bR_{>0}$}
            \Procedure{BudgetedCoverage}{$L \in \bR, ~\alpha$}
                \State $Y \leftarrow \setst{ x \in [n] }{ w_x \leq L }$
                \State $X_0 \leftarrow \emptyset$, $i \leftarrow 0$
                \While{$Y$ is non-empty}
                    \State Pick any $e_{i+1} \in \argmax_{x \in Y} \card{S_x \setminus S(X_i)}/w_x$
                    \State Remove $e_{i+1}$ from $Y$
                    \If{$w(X_i \union \set{e_{i+1}}) \leq \alpha L$}
                        \State Let $X_{i+1} \leftarrow X_i \union \set{e_{i+1}}$
                        \State Let $i \leftarrow i + 1$
                    \EndIf
                \EndWhile
                \State $X_g \leftarrow X_i$
                \State Pick any $x^* \in \argmax_{w_x \leq L} \card{S_x \setminus S(X_g)}$
                \State \textbf{return} $X_f \leftarrow X_g \union \set{x^*}$
            \EndProcedure
            \State Let $\cX$ be an empty list
            \State Let $\wmin \leftarrow \min_x w_x$ and $\wmax \leftarrow \max_x w_x$ 
            \ForAll{values $L$ of the form $(1+\eps)^i$ in the interval $[\wmin, (1+\eps)n \, \wmax]$}
                \State Let $X \leftarrow \textsc{BudgetedCoverage}(L, \ln n)$
                \State If $S(X)=U$, add $X$ to the list $\cX$
            \EndFor
            \State \textbf{return} an $X \in \cX$ that minimizes $w(X)$
        \EndProcedure
    \end{algorithmic}
\end{algorithm}

\comment{
\begin{itemize}
\item Let $X_0 \leftarrow \emptyset$, $Y \leftarrow [m]$, $i \leftarrow 0$
\item While $Y$ is non-empty
\begin{itemize}
    \item Pick $e_{i+1} \in \argmax_{x \in Y}
                    \card{S_x \setminus S(X_i)}/w_x$
    \item Remove $e_{i+1}$ from $Y$
    \item If $w_{e_{i+1}} + \sum_{x \in X_i} w_x \leq \alpha L$
    \begin{itemize}
        \item Let $X_{i+1} \leftarrow X_i \union \set{e_{i+1}}$
        \item Let $i \leftarrow i + 1$
    \end{itemize}
\end{itemize}
\item Pick $x^* \in \argmax_{x \in [m]} \card{S_x \setminus X_i}$
\item Output $X_f = X_i \union \set{x^*}$.
\end{itemize}
}

\begin{proof}
Obviously $\wmin \leq w(OPT) \leq n \wmax$,
and so one of the values $L$ used in \textsc{ApproxSC2}
will satisfy $w(OPT) \leq L \leq (1+\eps) w(OPT)$.
Let us call this value $L^*$.

Consider the invocation of \textsc{BudgetedCoverage} with $L=L^*$,
and let $X_f$ denote its output.
Note that 
$$
w(X_f) ~=~ w(X_g) + w_{x^*} ~\leq~ \alpha L^* + L^* ~\leq~ (1+\eps)(1+\ln m) w(OPT),
$$
since \textsc{BudgetedCoverage} enforces the constraint $w(X_i) \leq \alpha L$,
and $x^*$ was selected to have $w_{x^*} \leq L$.

It remains to show that $S(X_f)=U$.
Following Khuller et al.,
consider the first iteration of \textsc{BudgetedCoverage} in which a member of $OPT$ was removed
from $Y$ but not added to $X_{i+1}$, and let $\ell$ be the value of $i$ during that iteration.
If there is no such iteration, then $X_f$ is a superset of $OPT$,
so it trivially holds that $S(X_f) = U$.
The remainder of the analysis is performed at the time of this iteration.
For notational convenience we define $X_{\ell+1} = X_\ell \union \set{e_{\ell+1}}$,
even though the algorithm did not actually define $X_{\ell+1}$ in this way.

\comment{
\begin{claim}[Khuller et al.~\cite{KMN99}]
\ClaimName{KMN1}
For all $i = 1,\ldots,\ell+1$, we have
$$\card{S(X_i)} - \card{S(X_{i-1})}
    ~\geq~ \frac{w_{e_i}}{L} \big( \card{S(OPT)} - \card{S(X_{i-1})} \big)
$$
\end{claim}

\begin{proof}
\begin{align*}
&
\card{ S(OPT) } - \card{ S(X_{i-1})} \\
~\leq~&
\card{ S(OPT) \setminus S(X_{i-1})} \\
~=~&
\card{ S(OPT \setminus X_{i-1}) \setminus S(X_{i-1})} \\
~\leq~&
\frac{ L }{ \sum_{j \in OPT \setminus X_{i-1}} w_j }
     \cdot {\textstyle \card{ \Union_{j \in OPT \setminus X_{i-1}} S_j \setminus S(X_{i-1})} }
     \qquad\text{(since $\sum_{j \in OPT} w_j \leq L$)}\\
~\leq~&
L \cdot \frac{ \sum_{j \in OPT \setminus X_{i-1}} \card{S_j \setminus S(X_{i-1})} }
     { \sum_{j \in OPT \setminus X_{i-1}} w_j } \\
~\leq~&
L \cdot \max_{j \in OPT \setminus X_{i-1}} \frac{ \card{S_j \setminus S(X_{i-1})} }{ w_j } \\
~\leq~&
L \cdot \frac{ \card{S_{e_i} \setminus S(X_{i-1})} }{ w_{e_i} }
    \qquad\text{($e_i$ was the greedy choice to add to $X_{i-1}$, and $Y \supseteq OPT \setminus X_{i-1}$)} \\
~=~&
L \cdot \frac{ \card{S(X_i)} - \card{S(X_{i-1})} }{ w_{e_i} } 
\end{align*}
The assumption that $i \leq \ell+1$ is used to conclude that $Y \supseteq OPT \setminus X_{i-1}$.
\end{proof}
}

\begin{claim}[Khuller et al.~\cite{KMN99}]
\ClaimName{KMN2}
For all $i = 0,\ldots,\ell+1$,
$$\card{S(X_i)} \geq \Big(1 - \prod_{k=1}^i \big(1-\frac{w_{e_k}}{L}\big) \Big) \cdot \card{S(OPT)}.$$
\end{claim}

\comment{
\begin{proof}
By induction, the case $i=0$ being trivial.
So suppose $i \geq 1$.
\begin{align*}
\card{S(X_i)}
    &~=~ \card{S(X_{i-1})} + \big( \card{S(X_i)} - \card{S(X_{i-1})} \big) \\
    &~\geq~ \card{S(X_{i-1})} + \frac{w_{e_i}}{L} \big( \card{S(OPT)} - \card{S(X_{i-1})} \big)
        \qquad\text{(by \Claim{KMN1})}\\
    &~=~ \Big( 1 - \frac{w_{e_i}}{L} \Big) \card{S(X_{i-1})} + \frac{w_{e_i}}{L} \card{S(OPT)} \\
    &~\geq~ \Big( 1 - \frac{w_{e_i}}{L} \Big) \Big(1 - \prod_{k=1}^{i-1} (1-\frac{w_{e_k}}{L}) \Big) \card{S(OPT)}
        + \frac{w_{e_i}}{L} \card{S(OPT)}
        \qquad\text{(induction)}\\
    &~=~ \Big(1 - \prod_{k=1}^{i} (1-\frac{w_{e_k}}{L}) \Big) \card{S(OPT)}.
\end{align*}
The last line here is using the identity $(1-a)(1-b)+a = 1-(1-a)b$.
\end{proof}
}

\begin{fact}
\FactName{convexity}
Fix $\tau \leq t$.
Let $\Lambda = \setst{ \lambda \in [0,1]^t }{ \sum_{k=1}^t \lambda_k \geq \tau }$.
The function $1-\prod_{k=1}^t (1-\lambda_k)$ is minimized over $\Lambda$ 
by the point $\lambda^*$ with $\lambda^*_k = \tau/t$ for all $k$.
\end{fact}

By \Claim{KMN2}, we have
$$
\card{S(X_{\ell+1})}
    ~\geq~ \Bigg(1 - \prod_{k=1}^{\ell+1} \Big(1-\frac{w_{e_k}}{L}\Big) \Bigg) \cdot \card{S(OPT)}.
$$
Since $e_{\ell+1}$ was not added to the greedy solution, $\sum_{k=1}^{\ell+1} w_{e_k} \geq L \ln n$.
Since each $w_{e_k} \leq L$, we must have $\ell+1 \geq \ln n$.
Therefore we may apply \Fact{convexity} with
$\lambda_k = w_{e_k}/L$, $\tau = \ln n$, and $t=\ell+1$, to obtain
\begin{align*}
\card{S(X_{\ell+1})}
    ~\geq~ \Bigg(1 - \Big(1-\frac{\ln n}{\ell+1}\Big)^{\ell+1} \Bigg) \cdot \card{S(OPT)} 
    ~\geq~ (1 - 1/n) \cdot \card{S(OPT)}.
\end{align*}
As $\card{S(OPT)} \leq n$, we have $\card{S(X_{\ell+1})} \geq \card{S(OPT)}$.
%
Due to our artificial definition of $X_{\ell+1}$, it is not true that $X_{\ell+1} \subseteq X_f$.
However, since $X_\ell \subseteq X_g$ and by choice of $x^*$,
$$
\card{S(X_f)}
    ~=~ \card{S(X_g \union \set{x^*})} 
    ~\geq~ \card{S(X_g \union \set{e_{\ell+1}})} 
    ~\geq~ \card{S(X_\ell \union \set{e_{\ell+1}})} 
    ~=~ \card{S(X_{\ell+1})} 
    ~\geq~ \card{S(OPT)}.
$$
Thus, the invocation of \textsc{BudgetedCoverage} with $L=L^*$
outputs $X_f$ with $S(X_f)=U$ and $w(X_f) \leq (1+\eps)(1+\ln n) w(OPT)$.
\end{proof}

\subsection{MapReduce implementation}

Now we explain how how to design an MapReduce algorithm based on \textsc{ApproxSC2}.
The main observations are that
\begin{itemize}
\item the ``budgeted coverage problem''
$\max \setst{ \card{S(X)} }{ w(X) \leq L }$ is an instance of submodular maximization
over a hereditary set system.
\item \textsc{BudgetedCoverage} is a greedy-like algorithm satisfying certain
consistency conditions.
\end{itemize}
Consequently, an efficient MapReduce algorithm for \textsc{BudgetedCoverage}
can be obtained using the framework of Barbosa et al.~\cite{BENW16}.

The consistency property needed by Barbosa et al.\ is as follows.
First, we will assume that the $\argmax$ operations in \textsc{BudgetedCoverage}
follow an arbitrary but fixed tie-breaking rule.
Next, let $\textsc{BC}(A)$ denote the output of the \textsc{BudgetedCoverage} algorithm
(for some particular parameters $L$ and $\alpha$)
when provided only the sets $\setst{ S_x }{ x \in A }$.

\begin{claim}
Let $U$ and $V$ be disjoint subsets of $[n]$.
Suppose that $\textsc{BC}(U \union \set{v}) = \textsc{BC}(U)$ for all $v \in V$.
Then $\textsc{BC}(U \union V) = \textsc{BC}(U)$.
\end{claim}
\begin{proof}[Proof sketch]
Inductively one can show that
$\textsc{BC}(U)$,
$\textsc{BC}(U \union \set{v})$, and
$\textsc{BC}(U \union V)$ all build the same intermediate set $X_i$,
and so terminate with the same set $X_f$.
Since all have the same state $X_i$,
at the time that $v$ is considered by 
$\textsc{BC}(U \union \set{v})$ and $\textsc{BC}(U \union V)$,
they must behave identically, and so neither of them will add $v$ to the solution.
\end{proof}

\comment{
\begin{proof}
Suppose $\textsc{BC}(U \union V)$ adds $e_{i+1}$ to $X_{i+1}$ but $\textsc{BC}(U)$
does not add the same element; consider the first time that this happens.
At this point $X_i$ is the same in both executions.
We will see that all cases lead to a contradiction.

\textit{Case 1:}
If $e_{i+1} \in U$, then $\textsc{BC}(U)$ draws the same element,
so in both executions the comparison $w(X_i \union \set{e_{i+1}}) \leq \alpha L$ must also be the
same, and therefore the decision to add $e_{i+1}$ must also be the same.
Contradiction.

\textit{Case 2:}
If $e_{i+1} \in V$, then the execution $\textsc{BC}(U \union \set{v_k})$
also has the same set $X_i$, and also draws this element $e_{i+1}$.
The claim's hypothesis tells us that the execution $\textsc{BC}(U \union \set{v_k})$
does not add the element $v_k$.
But, in both executions the comparison $w(X_i \union \set{e_{i+1}}) \leq \alpha L$ must also be the
same. Contradiction.

Thus $X_g$ must be the same set in all executions
$\textsc{BC}(U)$, $\textsc{BC}(U \union V)$, and $\textsc{BC}(U \union \set{v_k})$ for all $k$.
Since $v_k$ was not chosen as the element $x^*$ in the execution
$\textsc{BC}(U \union \set{v_k})$, then there is some element $u \in U$
that is preferable as $x^*$ to all $v_k$ according to the $\argmax$ computation.
This same element $u$ is available to all executions
$\textsc{BC}(U)$, $\textsc{BC}(U \union V)$, and $\textsc{BC}(U \union \set{v_k})$ for all $k$.
It follows that the choice of $x^*$ is the same in all executions.
\end{proof}
}

The framework of Barbosa et al.\ yields a MapReduce implementation of 
\textsc{BudgetedCoverage}...

\subsubsection{Space usage}

There is a small issue regarding space usage (that does not arise with unweighted set cover).
Suppose that \textsc{ApproxSC2} outputs a set $X$.
Observe that $\card{X}$ can be arbitrarily bigger than $\card{OPT}$,
even though we have the guarantee $w(X_f)/w(OPT) = O(\ln n)$.

To address this issue, imagine for now that we know the values $\card{OPT}$ and $w(OPT)$.
Let us define $\beta = \eps w(OPT)/\card{OPT}$ and
apply \textsc{ApproxSC2} with modified weights $w_i' = w_i + \beta$.
Note that $w'(OPT) \leq w(OPT) + \beta \card{OPT} = (1+\eps) w(OPT)$.
So, for some value $L \leq (1+\eps) w'(OPT) \leq (1+\eps)^2 w(OPT)$, the algorithm
will succeed in finding a set $X_f$ with $S(X_f) = U$
and $w'(X_f) \leq (1+\eps)(1+\ln n) L$, and hence $w(X_f) \leq (1+\eps)^3 (1+\ln n) w(OPT)$.
This leads to a bound on 
$$
\card{X_f}
    ~\leq~ \frac{w'(X_f)}{\beta}
    ~\leq~ \frac{(1+\eps)(1+\ln n) L}{\eps w(OPT)/\card{OPT}}
    ~\leq~ (1+\eps)^3 (1+\ln n) \card{OPT} / \eps.
$$
`
We can try all pairs of powers-of-two $(K,L)$,
$K \in [1,n]$ and $L \in [\wmin, (1+\eps) n \wmax]$,
where $K$ serves as an upper bound for $\card{OPT}$
and $L$ serves as an upper bound for $w(OPT)$.
We may try all values of $L$ in parallel.
Unfortunately it seems we need to do ``galloping binary search'' to
find $K$ with $K = O(\card{OPT})$, but with out using too much memory.
}

\section{Maximum weight $b$-matching}
\AppendixName{bmatching}
\subsection{Local ratio method}

\begin{tcolorbox}[title=Sequential local ratio algorithm for maximum weight $b$-matching]
	Arbitrarily select an edge $e = (u,v)$ with positive weight, say $w(e)$.
    Reduce the weight of $e$ by $w(e)$.
    For all other $e' \ni u$, reduce its weight by $w(e)/b(u)$.
    For all other $e' \ni v$, reduce its weight by $w(e)/b(v)$.
    Push $e$ onto a stack and repeat this procedure until there are no positive weight edges remaining.
    At the end, unwind the stack adding edges greedily to the matching.
\end{tcolorbox}
\begin{theorem}
	Let $b = \max_v b(v)$.
    The above algorithm returns a $(3 - 2/\max\{2,b\})$-approximate matching.
\end{theorem}
\begin{proof}
   We assume $b \geq 2$ since the case $b = 1$ is exactly Theorem~\ref{thm:lr_matching}.

   Let $M_0 = \emptyset$ and, for $i \geq 1$, let $M_i$ be the matching maintained after we have unwinded $i$ edges from the stack.
   Next, let $G_0$ be the (weighted) graph just before we begin unwinding the edges and, for $i \geq 1$, let $G_i$ be the graph with the last $i$ weight reductions reversed.
   (So, $G_0, G_1, \ldots$ all have the same vertex and edge sets; only the edge weights are different.)

   We claim that $M_i$ is a $(3 - 2/b)$-approximate matching to $G_i$ for all $i = 0, 1, 2, \ldots$ which proves the claim.
   The case $i = 0$ is trivial since $G_0$ has no positive edge weights, so an empty matching is an optimal matching.

   Now suppose that $M_{i-1}$ is a $(3-2/b)$-approximate matching to $G_{i-1}$.
   Let $\OPT_{i-1}, \OPT_i$ denote an optimal matching of $G_{i-1}, G_i$, respectively and $w_{i-1}, w_i$ be the respective weight functions.
   Let $e_i = (u_i, v_i)$ be the $i$th edge that was popped from the stack.
   Note that
   \begin{align*}
      w_i(\OPT_i) & \leq w_{i-1}(\OPT_{i-1}) + w_i(e_i) + \left( \frac{b(u_i) - 1}{b(u_i)} + \frac{b(v_i) - 1}{b(v_i)} \right) w_i(e_i) \\
      & = w_{i-1}(\OPT_{i-1}) + w_i(e_i) \left( 3 - \frac{1}{b(u_i)} - \frac{1}{b(v_i)} \right) \\
      & \leq w_{i-1}(\OPT_{i-1}) + w_i(e_i) (3 - 2/b).
   \end{align*}
   Next, we give a lower bound for $w_i(M_i)$.
   Suppose first that we add $e_i$ to the matching, i.e.~$M_{i} = M_{i-1} \cup \{e_i\}$.
   Then
   \begin{align*}
      w_i(M_i) \geq w_{i-1}(M_{i-1}) + w_i(e_i).
   \end{align*}
   On the other hand, suppose we do not add $e_i$ to the matching, i.e.~$M_{i} = M_{i-1}$.
   In this case, $M_i$ must contain at least $b(u_i)$ edges incident to $e_i$ or $b(v_i)$ edges incident to $e_i$.
   This implies that $w_i(M_i) \geq w_{i-1}(M_{i-1}) + w_i(e_i)$ because undoing the weight reduction increases the weight of the edges incident to $u_i$ (other than $e_i$) by $w(e_i) / b(u_i)$
   and similarly for $v_i$.
   In either case, we have
   \begin{align*}
      (3 - 2/b) w_i(M_i) & \geq (3 - 2/b) (w_{i-1}(M_{i-1}) + w_i(e_i)) \\
      & \geq w_{i-1}(\OPT_{i-1}) + (3 - 2/b) w_i(e_i) \\
      & \geq w_i(\OPT_i),
   \end{align*}
   where the second inequality used the assumption that $M_{i-1}$ is a $(3-2/b)$-approximate matching to $G_{i-1}$.
\end{proof}

\subsection{Randomized local ratio}
Unfortunately, we cannot translate the above local ratio method for $b$-matching to the MapReduce model in the same way that we did for matching.
The issue is as follows.
Consider a vertex and suppose that it has $b$ edges all of unit weight.
In the matching case (i.e.~$b = 1$) if any of these edges were chosen by the local ratio algorithm then every edge incident to the vertex would be killed off by the weight reduction step.
On the other hand if $b > 1$ then the weight of the remaining edges after looking at $t < b$ edges is $(1 - 1/b)^t > 0$.
In particular, we have only managed to kill all the edges after we have looked at all the edges!

To fix this, we will consider an $\eps$-adjusted local ratio method which is inspired by the algorithm of Paz and Schwartzman~\cite{PS17} for weighted matchings in the semi-streaming model.
As in \Theorem{matchingMR}, we maintain a variable $\phi(v)$ for each vertex $v$ which corresponds to the sum of the weight reductions incident to vertex $v$.
However, rather than killing an edge $(u,v)$ if $w_{(u,v)} \leq \phi(u) + \phi(v)$, we kill the edge if $w_{(u,v)} \leq (1+\eps)(\phi(u) + \phi(v))$.
This corresponds to applying a reduction with a multiplier of either $1$ or $1+\eps$.
It is not hard to show that this gives a $(3 - 2/\max\{2,b\} + 2\eps)$-approximate $b$-matching.
Indeed, the only change to the above proof is to replace the inequality $w_i(\OPT_i) \leq w_{i-1}(\OPT_{i-1}) + (3 - 2/b)w_i(e_i)$
with $w_i(\OPT_i) \leq w_{i-1}(\OPT_{i-1}) + (3 - 2/b + 2\eps)w_i(e_i)$.

Let us now give some intuition for the MapReduce implementation of the $\eps$-adjusted local ratio method.
Consider a fixed a vertex $v$ and suppose we sample approximately $2b(v)$ of its edges uniformly at random. 
In expectation, $b(v)$ of the sampled edges will have weight at least the median weight edge adjacent to $v$. By reducing all of these edges in the local ratio algorithm, we can effectively remove half of the edges adjacent to $v$. Similarly, by drawing $\tO(b(v)n^{\mem})$ edges, we decrease the degree of $v$ by a factor of $n^{-\mem}$.

To implement the sequential local ratio algorithm in MapReduce, we essentially use the same algorithm as the one for maximum matching, adding up to $b \log \left(\delta^{-1}\right)$ edges to the stack for each vertex, where $\delta = \eps/(1+\eps)$. To simplify the analysis of the algorithm, we randomly sample a fixed number of edges from each vertex instead of sampling with uniform probability from the graph.
\begin{algorithm}
    \caption{$(3-2/b+2\eps)$-approximation for maximum weight $b$-matching. Blue lines are centralized.}
    \label{alg:matching}
    \begin{algorithmic}[1] 
        \Procedure{ApproxBMaxMatching}{$G = (V,E)$}
        	\State $\delta \gets \eps/(1+\eps)$
        	\State $E_1 \gets E$, $d_{1}(v) \gets d(v)$, $i \gets 1$
            \State $S \gets \emptyset$ \Comment{Initialize an empty stack.}
            \While{$E_i \neq \emptyset$}
            	\For{each vertex $v \in V$}
                	\If{ $|E_i| < 2b \ln \left(\delta^{-1}\right) n^{1+\mem}$ }
                    	\State Let $E'_v$ be all edges in $E_i$ incident to $v$ \label{ln:bmatching_done}
                    \Else
                		\State Randomly sample $b(v) \ln \left(\delta^{-1}\right) n^{\mem}$ edges from $E_i$ incident to $v$ and add to $E'_v$ \label{ln:bmatching_sample}
                    \EndIf
                \EndFor
                
                \CFOR{each vertex $v \in V$}
                	\CSTATE $j \gets 1$
                	\CWHILE{$j \leq b(v) \ln \left(\delta^{-1}\right) $} \label{ln:bmatching_wr}
                      \CSTATET Let $e \in E_v'$ be the heaviest edge and
                      apply an $\eps$-adjusted weight reduction
                      \CSTATE Push $e$ onto the stack $S$
                      \CSTATE Remove $e$ from $E_v'$
                      \CSTATE $j \gets j+1$
                    \ENDWHILE
                \ENDFOR
                
                \State Let $E_{i+1}$ be the subset of $E_i$ with positive weights
                \State $d_{i+1}(v) \gets |\{e \in E_{i+1} : v \in e\}|$
                	\label{bmatching:dupdate}
                \State $i \gets i + 1$
            \EndWhile
            \CSTATE Unwind $S$, adding edges greedily to the $b$-matching $M$
            \State \textbf{return} $M$
        \EndProcedure
    \end{algorithmic}
\end{algorithm}

The analysis of the algorithm is similar to the proof for maximum matching.
Instead of repeating the entire proof, we show a variant of Lemma~\ref{lem:phases} below and note that all the other lemmas extend similarly.

\begin{lemma}
    Suppose $\eta = n^{1+\mem}$ for some constant $\mem > 0$.
	Let $\Delta_i = \max_v d_i(v)$.
    For $i > 2$ with probability at least
    $1 - \frac{n^2}{\delta} \cdot \exp(-n^{\mem/2} / 2)$, it holds that $\Delta_{i+1} \leq \Delta_{i} / n^{\mem/4}$.
\end{lemma}
\begin{proof}
   If $\card{E_i} < 2b \ln(\delta^{-1}) n^{1+\mem}$ then the local ratio algorithm is performed on the entire graph, and so $E_{i+1}=\emptyset$ and the lemma is trivial.
   So, assume that $\card{E_i} \geq 2b \ln(\delta^{-1}) n^{1+\mem}$.

	Let $k_v$ be the number of edges incident to $v$ with positive weight when we reach $v$ in the for loop in line \ref{ln:weight_reduction}.
    If $k_v \leq \Delta_i / n^{\mem/4}$ then we are done, so suppose $k_v > \Delta_i / n^{\mem/4}$.
    
    Define an edge to be heavy if it is one of the top $k_v / n^{\mem/4}$ heaviest edges.
    We claim that we sampled at least $b(v) \ln(1/\delta)$ distinct heavy edges w.h.p.
    To see this, first split the sampled edges into $b(v) \ln(1/\delta)$ groups each with $n^{\mem}$ edges.
    Then, using a union bound, the probability that we \emph{cannot} pick a distinct heavy edge from each group is bounded above by
    \begin{equation}
       \label{eqn:bmatching_prob}
       \begin{multlined}
          \sum_{t=1}^{b(v) \ln(1/\delta)} \left( 1 - \frac{k_v / n^{\mem/4} - t}{\Delta_i} \right)^{n^{\mem}} \\
          \leq
          b(v) \ln(1/\delta) \left( 1 - \frac{\Delta_i / n^{\mem/2} - b(v) \ln(1/\delta)}{\Delta_i} \right)^{n^{\mem}}.
      \end{multlined}
    \end{equation}
    By assumption, we have $|E_i| \geq 2b(v) \ln \left(\delta^{-1}\right) n^{1+\mem}$,
    so since maximum degree is at least average degree, $\Delta_i/n^{\mem/2} \geq 2b(v) \log(1/\delta)$.
    Hence, \eqref{eqn:bmatching_prob} is at most
    \[
       b(v) \ln(1/\delta) \left( 1 - \frac{1}{2n^{\mem/2}} \right)^{n^{\mem}} \leq b(v) \ln(1/\delta) \exp\left( -n^{\mem/2} / 2 \right).
    \]
    So we have sampled at least $b(v) \ln(1/\delta)$ heavy edges
    with probability at least $1 - b(v) \ln(1/\delta) \exp\left( -n^{\mem/2} / 2 \right)$.
    We will condition on this event for the rest of the proof.
    
    Recall that in the algorithm we add the top $b(v) \ln(1/\delta)$ edges to the stack (note that the ordering of the remaining weights are unchanged in the while loop starting at Line~\ref{ln:bmatching_wr} because we subtract the same quantity from each edge).
    Let $w_L$ be the weight of the lightest of these edges \emph{before} we perform the weight reductions.
    Imagine for now that the algorithm performs ordinary weight reductions rather than $\eps$-adjusted reductions.
    Then, after performing the weight reductions, the weight of all non-heavy edges is at most $w_L (1 - 1/b(v))^{b(v) \ln(1/\delta)} \leq w_L \delta$.
    Since we choose $\delta = \eps/(1+\eps)$, this implies that we have reduced the weight of all non-heavy edges by a $1/(1+\eps)$ fraction of their original weight.
    So, since the algorithm actually performs $\eps$-adjusted weight reductions, it actually reduces the weight of the non-heavy edges to a non-positive value.
    These edges can now be safely discarded from the graph, so we have $d_{i+1}(v) \leq k_v/n^{\mem/4}$ after the end of the weight reduction phase (line \ref{bmatching:dupdate}). Taking a union bound completes the proof.
\end{proof}

\begin{theorem}
\label{thm:bmatch}
There is a MapReduce algorithm that computes a $(3-2/b+2\eps)$-approximation to the maximum weight $b$-matching in $O(c/\mem)$ rounds when $\mem > 0$ and $O(\log n)$ rounds when $\mem = 0$.
The memory requirement is $O(b \log(1/\eps) n^{1+\mem})$.
\end{theorem}

\comment{
\section{Older content: Minimum weight vertex cover}
Let $G = (V,E,w)$ be a graph where $w \colon V \to \R_{\geq 0}$ is a weight vector on the vertices.
A \emph{vertex cover} is a set $U \subseteq V$ such that every edge in $E$ is incident to at least one vertex in $U$.
A minimum weight vertex cover is a vertex cover $U$ that minimizes $\sum_{v \in U} w_v$.

\subsection{Local ratio method}
The local ratio algorithm for minimum vertex cover is given below.
\begin{tcolorbox}[title=Basic algorithm for minimum weight vertex cover]
	Arbitrarily select an edge $e$ whose vertices have positive weight and
    reduce the minimum of the two vertex weights from the two vertices.
    Remove all vertices with weight 0 from the graph and add them to the cover.
    Repeat so long as edges remain.
\end{tcolorbox}
The weight reduction step works as follows.
If the edge $e = (u,v)$ was selected and the vertex weights are $w_u$ and $w_v$, respectively,
then the weight reduction makes the updates $w_u \gets w_u - \min\{w_u, w_v\}$ and $w_v \gets w_v - \min\{w_u, w_v\}$.
\begin{theorem}
	\label{thm:lr_vc}
	The above algorithm gives a $2$-approximation to the minimum weight vertex cover problem.
\end{theorem}

\begin{proof}
	The output of the algorithm gives a valid vertex cover.
    To see this, suppose that an edge was not in the cover.
    Then both of the endpoints must have positive weight (since we output all nodes with zero weight).
    This is impossible, as the algorithm does not terminate so long as there are edges with both vertices having positive weight.

	We now prove that the vertex cover that is output by the algorithm is indeed a 2-approximate vertex cover.
	The linear program relaxation for minimum weight vertex cover is
    \begin{equation}
      \label{eqn:vc_p}
      \min \sum_{v \in V} w_v x_v \quad \text{s.t.} \quad x_u + x_v \geq 1 \,\, \forall (u,v) \in E, \quad x \geq 0.
   \end{equation}
   Taking the dual of \eqref{eqn:vc_p} gives
   \begin{equation}
      \label{eqn:vc_d}
      \max \sum_{e \in E} y_e \quad \text{s.t.} \quad \sum_{e \ni v} y_e \leq w_v \,\, \forall v \in V, \quad y \geq 0.
   \end{equation}
   For the proof, we can modify the basic algorithm to construct a feasible dual solution as follows.
   We initialize $y$ to be the $0$ vector.
   Consider the $t$-th iteration of the algorithm and suppose that we select the edge $e = (u,v)$.
   Let $w'_u, w'_v$ denote the current weights of the vertices $u, v$, respectively.
   Then we will increment $y_e$ by $\min\{w'_u, w'_v\}$ and decrement both $w'_u, w'_v$ by the same amount
   ($y_e'$ does not change for $e' \neq e$).
   To see that $y$ is feasible, we claim that at any iteration $t$ we have $w'_v + \sum_{e \ni v} y_e = w_v$.
   Indeed, this is true because whenever we select an edge $e$ incident to $v$, increment $w'_v$ and $y_e$ by the same amount.
   Finally, let $U$ be the set of vertices with weight $0$ at the end of the algorithm.
   Then $U$ is a vertex cover, as mentioned above, and since $y$ is feasible we have
   \[
      2\cdot \OPT \geq 2 \sum_{e \in E} y_e = \sum_{v \in V} \sum_{e \ni v} y_e \geq \sum_{v \in U} \sum_{e \ni v} y_e = \sum_{v \in U} w_v.
   \]
   To conclude, we have that $U$ is a 2-approximate vertex cover to the minimum weight vertex cover problem.
\end{proof}

\subsection{Local ratio and random sampling}
Our algorithm for approximate minimum weight vertex is similar to the algorithm in \cite{LMSV11} for finding maximal matchings.
The key difference is that in each iteration we apply the local ratio method to ensure that the number of edges that we need to consider
goes down by a large factor in each iteration.
\begin{algorithm}
    \caption{2-approximation for minimum weight vertex cover}
    \label{alg:minvc}
    \begin{algorithmic}[1] 
        \Procedure{ApproxMinVC}{$G = (V,E)$}
        	\State $G_1 \gets G$, $E_1 \gets E$, $i \gets 1$
            \While{$E_i \neq \emptyset$}
            	\State Sample each edge in $E_i$ with probability
                	   $p = \min\left(1, \frac{2\eta}{|E_i|}\right)$ and add it to $E'$ \label{alg:minvcsample}
                \If{$|E'| > 6\eta$}
                	\State \textbf{Fail}
                \EndIf
                \State Run the local ratio method on $G' = (V, E')$ \label{alg:minvclocalratio}
                \State Let $I$ be the set of vertices with positive weight
                \State $E_{i+1} \gets E[I]$
                \State $i \gets i + 1$
            \EndWhile
            \State \textbf{return} all $v \in V$ with weight 0
        \EndProcedure
    \end{algorithmic}
\end{algorithm}

We will need the following technical lemma from \cite{LMSV11} which says that any subgraph of $G$ induced
by the vertices either contains few edges or contains an edge that was sampled in \Cref{alg:minvcsample} of the algorithm.
\begin{lemma}[\cite{LMSV11}]
	\label{lem:filter}
	Let $G = (V,E)$ be a graph and $E' \subseteq E$ be a set where each $e \in E$ is added to $E'$ with probability $p$.
    Let $I \subseteq V$ and let $E[I]$ denote the set of edges with both endpoints in $I$.
    Then, with probability at least $1 - e^{-n}$, either $|E[I]| \leq 2n/p$ or $E[I] \cap E \neq \emptyset$.
\end{lemma}
\begin{corollary}
	\label{cor:filter}
	Let $I$ be the set of vertices with positive weight after \Cref{alg:minvclocalratio} of \Cref{alg:minvc}.
    Then $E[I] \leq 2n/p$ with probability at least $1 - e^{-n}$.
\end{corollary}
\begin{proof}
	Suppose that $E[I] > 2n/p$.
    Then by Lemma~\ref{lem:filter}, there exists with high probability $e = (u,v) \in E[I]$ that was sampled in the sampling step.
    Hence, the local ratio method will set either $w(u), w(v)$ to 0 which contradicts that $I$ is the set of vertices with positive weight.
\end{proof}

\begin{theorem}
	Suppose $|E| = O(n^{1+c})$ and let $\eta = n^{1+\mem}$.
    Then \Cref{alg:minvc} terminates in $O( c / \mem )$ iterations and returns a 2-approximate vertex
    cover w.h.p. Furthermore, if $\mem = 0$ and $\eta = 40n$, \Cref{alg:minvc} terminates in $O( \log n )$ iterations.
\end{theorem}
\begin{proof}
	By Theorem~\ref{thm:lr_vc}, the vertices give a 2-approximation to the minimum vertex cover. It remains to bound the number of iterations.
	By a standard Chernoff bound, an iteration fails with probability at most $\exp(-n^{1+\mem}) \leq \exp(-n)$.
    Combining this with Corollary~\ref{cor:filter}, after iteration $i$,
    we have $|E_{i+1}| \leq |E_i| / n^{\mem} \leq n^{1 + c - i \mem}$ w.h.p. When $\mem = 0$, choosing $\eta = 40n$ gives $|E_{i+1}| \leq |E_i| / 2$ w.h.p. 
    Hence, the algorithm terminates and outputs a set of vertices in at most $\lceil c / \mem \rceil$
    iterations when $\mem > 0$ and $O(\log n)$ iterations when $\mem = 0$.
\end{proof}

\subsection{MapReduce implementation}

As a corollary, we get the following theorem in the MapReduce model.
\begin{theorem}
There is a MapReduce algorithm that computes a 2-approximation to the minimum weight vertex cover in
$O(c / \mem)$ rounds when $\mem > 0$ and $O(\log n)$ rounds when $\mem = 0$.
\end{theorem}
\begin{proof}
	It suffices to show that each iteration of the while loop within \Cref{alg:minvc} can be done in a constant number of iterations. Lines 4-9 of \Cref{alg:minvc} consists of three phases: sampling, local ratio, and pruning. 
    
    Given the number of edges (which we assume is available from the previous iteration), the sampling phase can be done in one mapping round on every while loop iteration except the first. Here the mapper simply receives an edge $e = (u,v)$ as input, sends it to the local ratio reducer with probability $p$, and with probability $1-p$ to a secondary passthrough reducer that simply returns its input.
    
    Clearly the local ratio phase can be done in one round on one machine in a single reduce round,
    since the number of sampled edges is $O(n^{1+\mem})$. This local ratio reducer then sends out a list of $O(n)$ vertices that remain alive.
    
    The pruning of edges on Lines 8 and 9 can be done in at most two more MapReduce rounds. Each mapper maps an edge $(u,v)$ to the reducer with key $u$. Additionally, for each node $u$ we notify the reducer with key $u$ whether node $u$ is alive or not. The reducers with dead vertices simply throw away their edges, while the reducers with live vertices simply returns their input. Doing this for one more round, mapping edges $(u,v)$ to the reducer with key $v$, completes the pruning step.
    
    Hence each iteration of the while loop takes at most 3 MapReduce rounds.
\end{proof}

%

}

\section{Auxiliary Results}
\begin{theorem}[Chernoff bound]
   \label{thm:chernoff}
   Let $X_1, \ldots, X_n$ be independent random variables such that $X_i \in [0,1]$ with probability 1. Define $X = \sum_{i=1}^n X_i$ and let $\mem = \E X$.
   Then, for any $\eps > 0$, we have
   \[
      \Pr[X \geq (1 + \eps) \mem ] \leq \exp\left( - \frac{\min\{\eps, \eps^2\} \mem}{3} \right).
   \]
\end{theorem}

\begin{theorem}[Hajnal-Szemer\'{e}di]
   \label{thm:hajnal}
   Every graph with $n$ vertices and maximum degree bounded by $k-1$ can be coloured using $k$ colours so that each colour class has at least $\lfloor n / k \rfloor \geq n/(2k)$ vertices.
\end{theorem}

\ifpdf
    \bibliographystyle{plainurl}
\else
    \bibliographystyle{plain}
\fi
\begin{small}
\bibliography{refs}
\end{small}

\end{document}